\documentclass[12pt,a4paper]{article}

\usepackage[centertags]{amsmath}
\usepackage{amsfonts,amsthm,amssymb}
\usepackage{amssymb}
\usepackage{amsmath}
\usepackage{graphicx}
\usepackage{booktabs}
\usepackage{appendix}
\usepackage[hmargin=2.6cm,vmargin=2.6cm]{geometry}
\usepackage{etoolbox}
\usepackage{tikz}
\usepackage{soul}
\usepackage{comment}
\usepackage{float}
\usepackage{setspace}
\usepackage{xcolor,framed}
\usepackage[colorinlistoftodos,prependcaption,textsize=tiny]{todonotes}
\usepackage{soul}   
\setstcolor{red}
\usepackage[round]{natbib}

\usepackage{hyperref}
\hypersetup{colorlinks,citecolor=blue}

\usepackage{threeparttable}
\usepackage{array}
\usepackage{booktabs}

\usepackage{blindtext}

\usepackage{amsmath}
\usepackage{amsthm}
\usepackage{amssymb}
\usepackage{mathtools}
\usepackage{bbold}

\DeclarePairedDelimiterX\set[1]\lbrace\rbrace{#1}

\DeclareMathOperator*{\E}{\mathbb{E}}
\DeclareMathOperator*{\R}{\mathbb{R}}

\newcommand{\ind}[1]{\mathbf{\mathbb{1}}_{#1}}

\DeclareMathOperator*{\argmax}{argmax}
\DeclareMathOperator*{\argmin}{argmin}

\newtheorem{theorem}{Theorem}
\newtheorem{definition}{Definition}
\newtheorem{lemma}{Lemma}
\newtheorem{proposition}{Proposition}

\newtheorem{corollary}{Corollary}

\usepackage{graphicx}
\usepackage{subfig}
\usepackage{caption}

\title{Reputation-based Persuasion Platforms\footnote{The work by Moshe Tennenholtz and Omer Madmon was supported by funding from the European Research Council (ERC) under the European Union’s Horizon 2020 research and innovation programme (grant agreement 740435). Itai Arieli gratefully acknowledges support from the Israel Science Foundation (grant agreement 2029464). We would like to thank the anonymous reviewers for their helpful comments. An extended abstract of an earlier version of this paper appeared at the proceedings of the 16th International Symposium on Algorithmic Game Theory (SAGT 23').}
}
\author{
    Itai Arieli\thanks{%
    Faculty of Data and Decision Sciences, Technion - Israel Institute of Technology,
    iarieli@technion.ac.il},
    Omer Madmon\thanks{%
    Faculty of Data and Decision Sciences, Technion - Israel Institute of Technology,
    omermadmon@campus.technion.ac.il},
    Moshe Tennenholtz\thanks{%
    Faculty of Data and Decision Sciences, Technion - Israel Institute of Technology,
    moshet@technion.ac.il}
}

\begin{document}
\maketitle

\begin{abstract}
In this paper, we introduce a two-stage Bayesian persuasion model in which a third-party platform controls the information available to the sender about users’ preferences. We aim to characterize the optimal information disclosure policy of the platform, which maximizes average user utility, under the assumption that the sender also follows its own optimal policy. We show that this problem can be reduced to a model of market segmentation, in which probabilities are mapped into valuations. We then introduce a repeated variation of the persuasion platform problem in which myopic users arrive sequentially. In this setting, the platform controls the sender's information about users and maintains a reputation for the sender, punishing it if it fails to act truthfully on a certain subset of signals. We provide a characterization of the optimal platform policy in the reputation-based setting, which is then used to simplify the optimization problem of the platform.

\end{abstract}



\newpage

\section{Introduction}
\label{sec:intro}

\paragraph{Persuasion platforms} Bayesian persuasion refers to a situation where a sender (such as a seller) attempts to influence the decision of a receiver (such as a buyer) by presenting them with information. One classic example is a seller trying to sell a product of uncertain quality to a buyer. The seller may have some aggregate knowledge about the preferences and likelihood of purchasing the product for a group of buyers, but may not be able to distinguish among individual buyers. In contrast, other parties, such as online selling platforms like Amazon and eBay, may have access to more specific information about each buyer and their likelihood of purchasing the product. These platforms can use this information to reveal relevant characteristics of individual buyers to the seller, improving the efficiency of the persuasion process. Our work aims to study this double information disclosure mechanism in a persuasion setting.

In particular, we study this model under the assumption that the receiver (or the utility of the receiver) is drawn from a given commonly known distribution and in parallel independently the quality of the sender's product is determined. The quality of the product is known only to the sender and the user's utility is known only to the platform. The platform and the sender commit to their information revelation strategies at the outset. In the first stage, the platform reveals information to the sender about the receiver, and thereafter the sender reveals information to the receiver as a function of the information it received from the platform. We study this problem under the assumption that the platform tries to maximize the expected utility of the receivers.

Unlike the standard Bayesian persuasion setting, in our model, we have a two-stage process of information revelation. We study this model in two different settings: a one-shot setting and a repeated setting.  In the repeated setting, which is inspired by reputation considerations, reciprocity and punishments play a significant role.

We consider the standard product adoption setting with a binary set of states of the world (e.g., high-quality product and low-quality product) and a binary set of receiver actions (e.g., buy and do not buy the product). The sender only knows a prior distribution of user types, and the platform knows the realized user type. After the sender receives information from the platform it forms a posterior on the user type. Then, the state of the world (product quality) is drawn from a commonly known prior distribution, and is observed only by the sender. The sender then recommends according to the committed persuasion strategy and the receiver forms a posterior probability on the product quality and decides whether to buy or not. Lastly, the sender and receiver obtain their payoffs based on the receiver's action and the product quality. 
\footnote{Note that the terms 'receiver' and 'user' are used interchangeably.}

Our problem, then, is to find the subgame perfect equilibrium that is optimal for the platform, i.e., an information disclosure policy for the platform that maximizes the users' utility. We start with analyzing the one-shot setting. In this setting, we demonstrate how our model can be reduced to the seemingly unrelated market segmentation model proposed by \cite{bergemann2015}. In their model, a monopolistic producer determines an optimal price of a product in a given market, which is a distribution over consumers’ valuations.  \cite{bergemann2015}. provide an efficient algorithm for finding a segmentation of the initial market into multiple markets, such that the consumers’ utility is maximized assuming the producer determines an optimal price at each segment separately. In their terminology, segmentation is essentially a distribution over markets, such that its expectation is the initial market. The reduction from the persuasion platform problem to the market segmentation problem is done by mapping the persuasion thresholds of the users (i.e., the probabilities of lying to users regarding the product quality) into valuations in the market segmentation problem. We also show that under some conditions, the equivalence between the two problems can be extended beyond the case of binary state space.

\paragraph{Reputation-based persuasion platforms} Next, we introduce a repeated variation of the persuasion platform problem, in which myopic users arrive sequentially. The one-shot interaction between the platform, the sender, and the users is similar to before, except now the platform can not only control the sender’s user information, but it also maintains a reputation state of the sender. That is, the platform requires the sender to act truthfully on some signals and punishes the sender if it does not respect the requirement by permanently suspending the sender from the market and transmitting it from a high reputation to a low reputation, for all subsequent time periods.

At the low reputation state, the sender obtains a known, constant utility that is lower than the monopolistic sender utility, i.e., the utility of the sender when it is provided with no user information from the platform and relies only on the prior user distribution. This punishment utility comes from an outside option, which is less desirable for the sender than the case in which the platform provides no user information at all.

The lower the utility of the low reputation state is, the larger the punishment power the platform has, meaning it can request a more truthful behavior while still being incentive-compatible from the perspective of the sender.

In this reputation-based framework, the platform is assumed to have mechanisms to assess product quality post-sale, through direct observations or indirect indicators like user reviews and return frequencies. For instance, platforms like Amazon and eBay allow buyers to leave reviews and ratings after a purchase. These reviews serve as a direct indicator of the buyer's satisfaction and the product's quality, which the platform and potential buyers use to gauge the reliability of a seller. 
Alternatively, other platforms have return and refund policies that are activated when a product fails to meet the advertised quality. The rate of returns and disputes can be an indirect measure of product quality and seller reliability, which platforms potentially monitor closely.

For the repeated setting, we provide a characterization of an optimal platform policy. This policy relies on the fact that in the high reputation state, it is optimal for the platform to adopt the segmentation which is similar to the one it provided in the one-shot case. We also demonstrate how in low dimensions the problem can be solved using an exhaustive search over the search space, and show how the platform's optimal policy becomes less effective and more discriminative as the sender is less patient.

\paragraph{Our contribution} The contribution of this paper is threefold: First, we define a novel model of persuasion platforms, which captures, for example, the strategic nature of recommendation systems and aims to provide a mechanism to increase users’ welfare on average. Second, we show a reduction from the one-shot persuasion platform to the market segmentation model of \cite{bergemann2015}. This reduction enables solving the persuasion platform problem in efficient runtime. Finally, we propose the reputation-based persuasion platform model and provide a characterization of the optimal platform policy. The characterization makes use of the algorithm of \cite{bergemann2015}, which is used also to solve the one-shot case.

\paragraph{Paper structure} The paper is organized as follows: in section \ref{sec:persuasion-platforms} we introduce and study the one-shot persuasion platform problem. In section \ref{subsec:pp-model} we present the persuasion platform model and the problem of finding an optimal platform policy, in section \ref{subsec:pp-ms} we discuss the market segmentation model of \cite{bergemann2015}, and in section \ref{subsec:pp-optimal-platforms} we show the reduction from the persuasion platform problem to the market segmentation problem. Section \ref{subsec:pp-beyond-two-states} presents several extensions beyond the case of binary state space.

Then, section \ref{sec:reputation-based-pp} discusses the repeated, reputation-based variation of the persuasion platform problem: in section \ref{subsec:rbpp-model} we introduce the model, and section \ref{subsec:rbpp-optimal-platform} deals with the characterization of an optimal platform policy. The proof of the main Theorem appears in section \ref{subsec:rbpp-proof}. A numerical example of approximation of an optimal policy in the reputation-based setting is given in section \ref{subsec:rbpp-numerical}. We then conclude in section \ref{sec:discussion}. Proofs of all technical lemmas are deferred to the appendix.

\section{Related Work}
\label{sec:related-work}

There is a rich literature on markets with asymmetric information. \cite{akerlof} introduced the very first model, which introduced the market for “lemons”, in which sellers have private information regarding the product quality. More recent models of information disclosure include the cheap talk (\citealp{crawford1982strategic, cheap-talk}) and Bayesian persuasion (\citealp{bayesian-persuasion}). In our model, we rely on the Bayesian persuasion framework, which differs from cheap talk by the fact that the sender has commitment abilities in addition to its private information. Our model describes a two-level persuasion process, in which a third-party platform controls the information the sender has about the users. 

There is also literature on multi-level, sequential Bayesian persuasion, e.g. \cite{li2021sequential,arieli2022bayesian, mahzoon2022hierarchical}. In contrast to these works, in our model the information that is revealed to the sender is on the user type, whereas the information that is revealed by the sender is on the product quality. In addition, we also extend the sequential persuasion setting to a repeated, reputation-based setting, in which sequential persuasion games are being played for an infinite number of time periods with myopic users, and characterize the optimal platform policy in terms of both information revelation and punishment strategy.

Signaling schemes are well-studied in various economic settings. For instance, \cite{signalling} have studied optimal signaling schemes in the context of a second-price auction with probabilistic goods, where optimality is measured with respect to the auctioneer's revenue. In contrast, our problem deals with maximizing the users' average utility. \cite{wisdom-of} have also studied a repeated dynamic persuasion model, in which the sender is aiming to persuade the receivers towards exploration in order to maximize social welfare. Notice that while they considered a single level of persuasion, in our model we study a two-level persuasion: from the platform to the sender and from the sender to the receiver.

We rely on the market segmentation model of \cite{bergemann2015} and utilize their algorithm to solve our one-shot persuasion platforms problem and characterize an optimal platform policy in the reputation-based persuasion platforms problem. Several extensions of the market segmentation model were studied recently, including buyers with incomplete information about their own types (\citealp{deb2021multidimensional,kartik2023lemonade}), robustness to the seller's type (\citealp{arieli2024robust}), segmentation under fairness and privacy constraints (\citealp{hidir2021privacy,cohen2022price,strack2023privacy,banerjee2023fair}), and more. The strong relationship between the persuasion platforms problem and the market segmentation problem implies that many of these extensions can be applied to the persuasion platforms problem as well.

In our reputation-based setting the platform also punishes the sender once it is caught lying to a user it should not have lied to according to the platform's request. 
The role of reputation in repeated games with incomplete information is well studied, e.g. \cite{reputation,monitoring}. In our reputation-based setting, we exploit a similar notion of reputation in an information design setting. In particular, \cite{platform-markets} studied the role of reputation in online platform markets, and highlighted empirical evidence for the fact that reputation systems often have a positive effect on online markets, in terms of both stability and efficiency. Related theoretical models of reputation and information design are considered by \cite{best2020persuasion} and \cite{mathevet2022reputation}. They used reputation in a repeated setting in a model that relaxes the commitment power assumption of the sender.

Other work considered dynamic Bayesian Persuasion models, in which a sender and a receiver interact repeatedly, and the state of the world evolves according to a Markovian law (\citealp{optimal-dynamic, markovian-persuasion}). Since in our model the state of the world is drawn independently every time period, our problem has a stationary optimal solution. \cite{markov-persuasion-process} considered a different Markovian setting, in which an informed sender is willing to persuade a stream of myopic receivers to take actions that maximize its cumulative utility. While they focus on the sender's optimal persuasion strategy, in this work we aim to characterize an optimal platform policy that maximizes the average receivers' utility.

\cite{data-collection} dealt with the receiver's welfare by presenting a sender-receiver Bayesian model in which the sender can request additional information from the receiver. While we are also interested in the receiver's welfare, we take a mechanism design perspective rather than allowing additional communication from the receiver to the sender.

\section{Persuasion Platforms}
\label{sec:persuasion-platforms}
\subsection{The Model}
\label{subsec:pp-model}

We denote by $\Theta = \{\theta_1, ... \theta_n\} \subset \R_{+}$ the set of users' types. We assume $0 < \theta_1 < ... < \theta_n$. $x^* \in \Delta(\Theta)$  is the prior user distribution. $\Omega = \{ \omega_0, \omega_1 \}$ is the set of states of the world, i.e. product qualities. We think of $\omega_0$ as the state of the world in which the product is of low quality, and $\omega_1$ is the state of the world in which the product is of high quality. $\mu \in \Delta(\Omega)$ is the prior distribution over product qualities. We identify $\mu = P_{\mu}(\omega_1)$. Let $A = \{ a_0, a_1 \}$ be the set of user actions (and sender recommendations), corresponding to buying ($a_1$) and not buying ($a_0$) the product. We denote the user's action by $a$ and the sender's recommendation by $\tilde{a}$.

A strategy of the platform is a signaling policy $\sigma: \Theta \to \Delta(S)$ for some abstract finite signal realizations space $S$. The platform's strategy can be thought of as a set of conditional distributions over $S$ (conditioned on the user type). Alternatively, a strategy of the platform is a Bayes-plausible distribution of distributions $\sigma \in \Delta(\Delta(\Theta))$. That is, $|\operatorname{supp}(\sigma)| < \infty$ and $\E_{x \sim \sigma} [x] = x^*$. We denote the set of all Bayes-plausible distributions by $\Sigma$.

A strategy of the sender is a recommendation policy conditioned on the signal realization and the state of the world, i.e. $p: \Omega \times S \to \Delta(A)$. We denote $p(\omega, s) = P_{p}(\tilde{a}=a_1 | \omega, s)$ for any $s\in S, \omega \in \Omega$. It is a well-known result that the optimal sender strategy always satisfies $p(\omega_1, s) = 1$ for any $s \in S$. Therefore, the sender's strategy can be solely characterized by its recommendation policy conditioned on the low-quality product, i.e. $p: S \to \Delta(A)$, where $p(s) = P_{p}(\tilde{a}=a_1 | \omega_0, s)$, and $P_{p}(\tilde{a}=a_1 | \omega_1, s)=1$. Note that the domain of $p$ can be defined as $\Delta(\Theta)$ instead of $S$, as each signal realization induces a posterior probability over the user types.

The sender has a utility function of $u_S(a_0) = 0, u_S(a_1) = 1$, and a user of type $\theta$ has a utility function of $u_R^{\theta}(a_0, \omega_0) = u_R^{\theta}(a_0, \omega_1) = 0, u_R^{\theta}(a_1, \omega_0) = -1, u_R^{\theta}(a_1, \omega_1) = \theta$. The platform's utility is the users' average utility (w.r.t the distribution of user distributions induced by its signaling policy $\sigma$).
The interaction between the three entities (platform, sender, and user) is then defined as follows:
\begin{enumerate}
    \item The platform commits to a strategy $\sigma$.
    \item The sender observes $\sigma$, and commits to its own strategy $p$.
    \item A user $\theta \sim x^*$ is drawn, and its type is visible only to the platform.
    \item The platform sends a signal $s \in S$ to the sender, according to the committed policy $\sigma$. That is, $s \sim \sigma(\theta)$.
    \item The state of the world $\omega \sim \mu$ is drawn, and it is visible to the sender only.
    \item If $\omega = \omega_1$, the sender sends a recommendation $\tilde{a} = a_1$ to the user. Otherwise, the sender sends a recommendation according to the committed policy $p$. That is, $\tilde{a} \sim p(s)$.
    \item The user observes $\tilde{a}$ and plays $a = a_1$ if and only if $\E_{\omega \sim \tilde{\mu}}[u_R^{\theta}(a_1, \omega)] \ge \E_{\omega \sim \tilde{\mu}}[u_R^{\theta}(a_0, \omega)]$, where $\tilde{\mu} \in \Delta(\Omega)$ is the posterior computed by the user based on $\mu$ and $\tilde{a}$ (otherwise it plays $a = a_0$). Simple algebra reveals that the condition holds for a user of type $\theta$ if and only if $p(s) \le \frac{\mu}{1 - \mu} \theta$. We refer to the quota $\frac{\mu}{1 - \mu} \theta$ as the persuasion threshold of a $\theta$-type user, and denote it by $\tau_{\theta}$.
\end{enumerate}

Figure \ref{gametree} visualizes the timing of the Platform, Sender, and User interaction. A user of type $\theta$ then plays according to the best-response mapping: 
\begin{center}
    $BR_{\theta}(\tilde{a}, p) = \begin{cases} \tilde{a} & p \le \tau_\theta \\ a_0 & p > \tau_\theta \end{cases}$
\end{center}
where $\tilde{a}$ is the received recommendation and $p = P(\tilde{a} = a_1 | \omega_0)$; The sender is maximizing: 
\begin{align}
U_S(\sigma, p) &= \E_{x \sim \sigma} [U_S(p;x)] \nonumber \\ 
&= \E_{x \sim \sigma} \E_{\theta \sim x} \E_{\omega \sim \mu} \E_{\tilde{a} \sim p(\omega, x)} [u_S(BR_{\theta}(\tilde{a}, p(\omega, x)))] \nonumber
\end{align}
and the platform is maximizing: 
\begin{align}
    U_P(\sigma, p) &= \E_{x \sim \sigma} [U_P(p;x)] \nonumber \\
    &= \E_{x \sim \sigma} \E_{\theta \sim x} [U_R^{\theta}(p;x)] = \E_{x \sim \sigma} \E_{\theta \sim x} \E_{\omega \sim \mu} \E_{\tilde{a} \sim p(\omega, x)} [u_R^{\theta}(BR_{\theta}(\tilde{a}, p(\omega, x)), \omega)] \nonumber
\end{align}


Note that we use the notation $U_S(p;x)$ for the sender's expected utility from playing $p$ conditional on the user posterior distribution $x$, and denote by $U_S(\sigma,p)$ the expected utility when the platform provides user information according to policy $\sigma$, which is taking the expectation over all $x \sim \sigma$ (and similarly we use $U_P(p;x)$ and $U_P(\sigma,p)$, to denote the expected utility of the platform). $U_R^{\theta}(p;x)$ is the expected utility of a $\theta-$type user when it responds optimally, as a function of the sender's strategy $p$ and the user distribution $x$. Note that by definition, the sender and platform utilities can be written explicitly as follows:

\begin{align}
U_S(\sigma, p) = \E_{x \sim \sigma} [U_S(p;x)] = \sum_{x \in \operatorname{supp}(\sigma)} \sigma(x) \sum_{j=1}^n x_j \ind{p(x) \le \tau_{\theta_j}} (\mu + (1 - \mu) p(x)) \nonumber \\
U_P(\sigma, p) = \E_{x \sim \sigma} [U_P(p;x)] = \sum_{x \in \operatorname{supp}(\sigma)} \sigma(x) \sum_{j=1}^n x_j \ind{p(x) \le \tau_{\theta_j}} (\mu \theta_j - (1 - \mu) p(x)) \nonumber
\end{align}
where $\ind{a \le b} = 1$ if $a \le b$ and $0$ otherwise. 
It is straightforward to see that any optimal sender strategy $p^*$ must satisfy that for every $x \in \Delta(\Theta)$, $p^*(x) \in \{\tau_\theta\}_{\theta \in \Theta}$.

    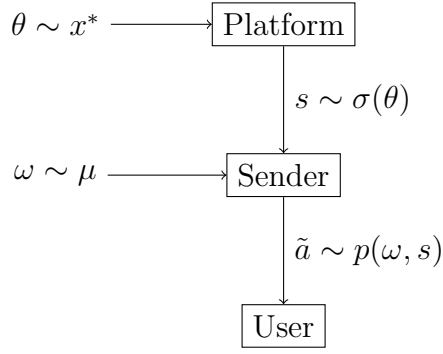
\begin{figure}
        \centering
        \begin{tikzpicture}[node distance=2cm]
            \node (platform) [rectangle, draw] {Platform};
            \node (sender) [rectangle, draw, below of=platform] {Sender};
            \node (user) [rectangle, draw, below of=sender] {User};
            \node (type) [left of=platform, xshift=-1cm] {$\theta \sim x^*$};
            \node (state) [left of=sender, xshift=-1cm] {$\omega \sim \mu$};

            \draw [->] (platform) -- (sender) node[midway, right] {$s \sim \sigma(\theta)$};
            \draw [->] (sender) -- (user) node[midway, right] {$\tilde{a} \sim p(\omega, s)$};
            \draw [->] (type) -- (platform);
            \draw [->] (state) -- (sender);
        \end{tikzpicture}
        \caption{The interaction between the Platform, the Sender, and the User.}
    \label{gametree}
    \end{figure}

\subsection{Market Segmentation}
\label{subsec:pp-ms}

We now describe a model for market segmentation, presented by \cite{bergemann2015}. In this model a monopolist producer sells a good to a continuum of consumers, where each consumer demands exactly one unit of the good. We denote by $V = \{v_1, ... v_n \}$ the set of consumer valuations for the good. We assume $0 < v_1 < ... < v_n$. A market $\pi \in \Pi = \Delta(V)$ is a distribution over the possible valuations.

In a given market $\pi$, the demand for the good at any price in the interval $(v_{k-1},v_k]$ is $\sum_{j=k}^n \pi_j$ (with the convention that $v_0=0$). A price $v_k$ is said to be optimal (for the producer) in a market $\pi$ if for all $i\in[n]$:  $v_k \sum_{j=k}^n \pi_j \ge v_i \sum_{j=i}^n \pi_j$. We denote by $\Pi_k$ the set of all markets where price $v_k$ is optimal. We hold a given initial (aggregate) market denoted by $\pi^* \in \Pi$. We denote by $v^* = v_{i^*}$ the optimal uniform price for the initial market $\pi^*$. Thus, $\pi^* \in \Pi^* = \Pi_{i^*}$.

Segmentation is a division of the aggregate market into different markets. Thus, a segmentation $\sigma$ is a finite-support distribution over markets, with the interpretation that $\sigma(\pi)$ is the proportion of the population in the market $\pi$. Thus, the set of possible segmentations is:
\begin{align}
    \Sigma = \set[\Big]{\sigma \in \Delta(\Pi) | \sum_{\pi \in \operatorname{supp}(\sigma)} \sigma(\pi)\pi = \pi^*, |\operatorname{supp}(\sigma)| < \infty}. \nonumber
\end{align}

A pricing rule for a segmentation $\sigma$ specifies a price for each segment (market in the support of $\sigma$). Formally, a pricing rule for a segmentation $\sigma$ is a function $\phi: \operatorname{supp}(\sigma) \to V$. A pricing rule $\phi$ is optimal if for each segment $\pi \in \operatorname{supp}(\sigma)$, $\phi(\pi) = v_k$ implies $\pi \in \Pi_k$. That is, the charged price in every market $\pi$ must be optimal for the producer.

We denote the utility of a consumer with valuation $v_j$ when charged price is $v$ by $W_j(v) = \ind{v \le v_j} (v_j - v)$, and the utility of a producer charging price $v$ in the market $\pi$ by $W_S(v;\pi) = v \sum_{j=1}^n \ind{v \le v_j} \pi_j$.

Given a segmentation $\sigma$ and a pricing rule $\phi$, we define the consumer surplus: 
\begin{align}
    W_C(\sigma,\phi) = \sum_{\pi \in \operatorname{supp}(\sigma)} \sigma(\pi) \sum_{j=1}^n \pi_j W_j(\phi(\pi)) = \sum_{\pi \in \operatorname{supp}(\sigma)} \sigma(\pi) \sum_{j=1}^n \pi_j \ind{\phi(\pi) \le v_j} (v_j - \phi(\pi)) \nonumber
\end{align}
and the producer surplus: 
\begin{align}
    W_S(\sigma,\phi) = \sum_{\pi \in \operatorname{supp}(\sigma)} \sigma(\pi) W_S(\phi(\pi);\pi) = \sum_{\pi \in \operatorname{supp}(\sigma)} \sigma(\pi) \phi(\pi) \sum_{j=1}^n \ind{\phi(\pi) \le v_j} \pi_j \nonumber
\end{align}

The market segmentation model can naturally be extended to the case where the determined price $v$ is not necessarily in the set of user types $V$ (that is, $\phi: \operatorname{supp}(\sigma) \to \R_{+}$), although it is clear that any optimal pricing rule must satisfy $\forall \pi \in \operatorname{supp}(\sigma): \phi^*(\pi) \in V$. \cite{bergemann2015} also characterizes the set of feasible producer and consumer surpluses, see Figure \ref{bbm-triangle}. 

\begin{figure}[h]
\centering
\begin{tikzpicture}[scale=0.6]
    \draw[-latex] (-1,0) -- (6,0) node[right] {$\text{Consumer surplus}$};
    \draw[-latex] (0,-1) -- (0,6) node[above] {$\text{Producer surplus}$};
    \coordinate[label=left:$A$] (A) at (0,1);
    \coordinate[label=right:$B$] (B) at (4,1);
    \coordinate[label=left:$C$] (C) at (0,5);
    \draw (A) -- (B) -- (C) -- cycle;
\end{tikzpicture}
\caption{The feasible surplus triangle. \cite{bergemann2015} provide an algorithm for finding a consumer surplus maximizing segmentation $\sigma^*$ for a given market $\pi^*$ (point B).}
\label{bbm-triangle}
\end{figure}
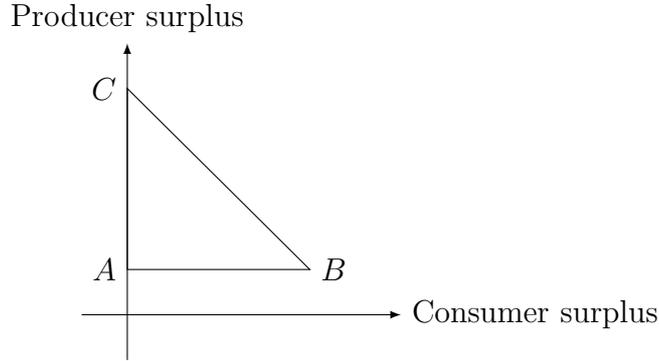

\subsection{Optimal Platform Policy: From Probabilities To Valuations}
\label{subsec:pp-optimal-platforms}

In this section, we show a reduction from the persuasion platform problem to the market segmentation problem. We then rely on the fact that \cite{bergemann2015} provides an algorithm for finding a consumer-surplus-maximizing segmentation in order to solve the platform persuasion problem. The reduction is done by defining an equivalent segmentation problem, in which the users' valuations for the products are the sender utilities from selling the product to users in the persuasion problem, and the initial (aggregate) market is the prior user distribution. We show that the games induced by the two problems are strategically equivalent, and therefore solving the equivalent segmentation problem imminently yields a solution for the platform persuasion problem.

\begin{definition}
    Given a persuasion platform problem instance $(\Theta, x^*, \mu)$, the equivalent market segmentation problem is the market segmentation problem instance $(V, \pi^*)$ defined as follows:
    
    \begin{equation}
        V = \set{\mu + (1-\mu) \tau_{\theta} }_{\theta \in \Theta} = \set{\mu (1 + \theta) }_{\theta \in \Theta} \nonumber 
    \end{equation}
        \begin{equation}
        \pi^* = x^* \nonumber
    \end{equation}
\end{definition}

Notice that in the equivalent market segmentation problem, the valuations of the consumers are derived from the users' persuasion thresholds, i.e. the probabilities of the sender recommending low-quality products which make the users indifferent between buying and not buying the product when the sender recommends. In other words, \emph{probabilities are translated into valuations.}
In the equivalent market segmentation problem, each user distribution $x$ is identified with a market $\pi$, and each sender policy $p$ is identified with a pricing rule of the producer $\phi(\pi) = \mu + (1 - \mu) p(x)$.
    
\begin{theorem}
\label{thm1}
Given a persuasion platform problem instance $(\Theta, x^*, \mu)$, consider the equivalent market segmentation problem $(V, \pi^*)$. For any user distribution $x$ (and its corresponding market $\pi$) and for any sender policy $p$ (and its corresponding pricing rule $\phi$) the following  properties hold:
\begin{enumerate}
    \item $U_S(p;x) = W_S(\phi(\pi);\pi)$
    \item $\forall j \in [n]$: $U_R^{\theta_j}(p;x) = W_j(\phi(\pi))$
\end{enumerate}
\end{theorem}

\begin{proof}
Let $x\in\Delta(\Theta)$ and a sender policy $p$. First, notice that $p(x) \le \frac{\mu}{1-\mu}\theta_j$ if and only if $\phi(\pi) \le v_j$ Since $f(q)=\mu + (1-\mu)q$ is monotonically increasing, and $\phi(\pi) = f(p(x))$. Therefore, from the definition of the sender and producer utilities, it follows that:
\begin{align}
U_S(p;x) = \sum_{j=1}^n x_j \ind{p(x) \le \frac{\mu}{1-\mu}\theta_j} (\mu + (1 - \mu) p(x)) = \sum_{j=1}^n \pi_j \ind{\phi(\pi) \le v_j} \phi(\pi) = W_S(\phi(\pi);\pi) \nonumber
\end{align}

Then, notice that if $\phi(\pi) \le v_j$ the utility of consumer of type $v_j$ from buying the product at price $\phi(\pi)$ is exactly:
\begin{align}
    W_j(\phi(\pi)) = v_j - \phi(\pi) &= \mu + (1 - \mu) (\frac{\mu}{1-\mu}\theta_j) - \mu - (1 - \mu) p(x) \nonumber \\
    &= (1 - \mu) (\frac{\mu}{1-\mu}\theta_j - p(x)) = \mu \theta_j - (1-\mu) p(x) = U_R^{\theta_j}(p;x) \nonumber
\end{align}
and otherwise they both equal $0$.
\end{proof}

Note that the equivalent market segmentation problem is indeed strategically equivalent to the persuasion problem only since we have restricted the sender's strategy space to strategies in which the sender always recommends buying a good product. The following corollaries, which follow immediately from Theorem \ref{thm1}, will be used to solve the problem of finding an optimal platform policy in a given persuasion platform problem:

\begin{corollary}
\label{cor1}
For every distribution $x \in \Delta(\Theta)$ (and the corresponding distribution $\pi \in \Delta(V)$) there exists some $k \in [n]$, such that the optimal sender strategy and the optimal pricing rule satisfy $\phi(\pi)=v_k$, $p^*(x)=\frac{\mu}{1-\mu}\theta_k$.
\end{corollary}

\begin{corollary}
\label{cor2}
For every Bayes-plausible distribution of user-distributions $\sigma \in \Sigma$, $U_S(\sigma,p) = W_S(\sigma,\phi)$ and $U_P(\sigma,p) = W_C(\sigma,\phi)$, where $p$ and $\phi$ are the best-responses to $\sigma$ at the persuasion and segmentation problems respectively.
\end{corollary}

\begin{corollary}
\label{cor3}
Let $\sigma^*$ be the segmentation that maximizes the consumer surplus at the equivalent market segmentation problem. Then $\sigma^*$ also maximizes the average user utility at the persuasion platform problem.
\end{corollary}

Corollary \ref{cor1} follows directly from the first property of Theorem \ref{thm1}. As for the Corollary \ref{cor2}, first note that the set of platform strategies is exactly $\Sigma$, since for every Bayes-plausible distribution $\sigma \in \Sigma$ there exists a signaling policy that induces it. The Corollary \ref{cor2} then follows from this fact combined with Corollary \ref{cor1}, and the second property of Theorem \ref{thm1}. Corollary \ref{cor3} follows directly from the Corollary \ref{cor2}.

We conclude this section by recalling that there exists an efficient algorithm for finding a consumer surplus maximizing segmentation:

\begin{theorem}{(\citealp{bergemann2015})}
    \label{thm2}
     There exists an algorithm that given a market segmentation problem instance $(V, \pi^*)$ finds a consumer surplus maximizing segmentation $\sigma^*$ in $O(|V|)$. Moreover, this segmentation satisfies $W_S(\sigma^*, \phi^*) = W_S(\phi^*(\pi^*); \pi^*)$ (where $\phi^*$ is an optimal pricing rule).
\end{theorem}

The fact that the reduction is done in linear time (combined with Theorem \ref{thm2} and Corollary \ref{cor3}) implies that there also exists an algorithm for finding an average user utility maximizing platform policy $\sigma^*$ in every persuasion platform problem instance $(\Theta, x^*, \mu)$. Moreover, this policy satisfies $U_S(\sigma^*, p^*)=U_S(p^*;x^*)$.

\paragraph{Achieving the Pareto-frontier} In practice, a two-sided platform may care about the welfare of both sellers and buyers, aiming to find a balance between the two. For example, the platform may charge fees from the sellers (which increase when the seller's revenue increases), or the platform wants to preserve the satisfaction of both the users and the sellers to avoid retention. Note that since the persuasion problem can be reduced to a problem of market segmentation, the set of feasible platform and sender utilities in the persuasion platform problem has a similar structure as in Figure \ref{bbm-triangle}. Therefore, solving the problem of maximizing the average user utility can be used to achieve any pair of sender and platform utilities on the Pareto-frontier: since optimal sender utility (point C in Figure \ref{bbm-triangle}) can be achieved by full revelation of the user type, and since every pair of utilities on the Pareto-frontier is a convex combination of C and B, one can achieve any such point by simply randomizing between the user optimal policy (derived from \citealp{bergemann2015}) and the full revelation policy.

\subsection{Beyond Binary State Space}
\label{subsec:pp-beyond-two-states}

In this section, we show how the reduction to the market segmentation problem, presented in Theorem \ref{thm1}, can be generalized beyond the binary state space case presented in our model. We provide a sufficient condition on the user utilities that guarantees the reduction indeed yields a solution for the persuasion platform problem, and then provide two concrete examples of natural cases satisfying the condition. 

We consider a generalized adoption persuasion problem, with a state space $\Omega$ that is an arbitrary abstract measurable state space equipped with a prior probability distribution $\mu \in \Delta(\Omega)$, and $\Theta$ is an abstract finite set of user types equipped with the prior probability distribution $x^* \in \Delta(\Theta)$. We assume that the prior over $\Omega \times \Theta$ is the product distribution $\mu \times x^*$. In addition, the receiver's action set $A=\{a_0,a_1\}$ is binary, and the receiver of type $\theta$ has a utility function $u^\theta_R: \Omega \times A \to \mathbb{R}$ such that $a_0$ represents the decision to opt-out and thus guarantees a constant utility of $0$ and $a_1$ represents a risky action of adopting the product.

The sender's utility function is still a state-independent utility that yields a utility of $1$ if and only if the action of the receiver is $a_1$, and zero otherwise.

Consider the persuasion problem where the prior over types $x^\theta\in\Delta(\Theta)$ assigns a probability one to the user of type $\theta$ for some $\theta\in\Theta$. That is, $x^\theta$ is the Dirac distribution over $\theta$ and represents the case where the sender knows that the receiver's type is $\theta$. We note that in this case there exists an optimal straightforward recommendation policy $p_\theta:\Omega\to\Delta(A)$ for the sender that is incentive compatible for the receiver of type $\theta$. We again abuse the notation and identify the image of $p_\theta$ with the probability of recommending the action $a_1$.

For any pair of types $\theta, \theta' \in \Theta$,  denote by $S_\theta$ the utility of the sender when playing optimally against a user of known type $\theta$, and denote by $R_{\theta \to \theta'}$ the utility of a user of type $\theta$ upon following the recommendation policy $p_{\theta'}$ (note that $p_{\theta'}$ is not necessarily incentive compatible for the receiver of type $\theta$).  Formally, we define:

\begin{equation*}
    S_\theta = \E_{\omega \sim \mu} [p_\theta(\omega)]
\end{equation*}
\begin{equation*}
    R_{\theta \to \theta'} = 
    \E_{\omega \sim \mu}  [p_{\theta'}(\omega)u_R^\theta(a_1,\omega)]
\end{equation*}

For instance, in our standard binary state space model, it holds that:

\begin{equation*}
    S_\theta = \mu (1 + \theta)
\end{equation*}
\begin{equation*}
    R_{\theta \to \theta'} = \mu (\theta - \theta')
\end{equation*}

We now make a similar reduction to a market segmentation problem by setting $V = \{ v_\theta \}_{\theta \in \Theta} = \{S_{\theta}\}_{\theta \in \Theta}$ and $\pi^* = x^*$. The following proposition provides sufficient conditions for the same reduction to solve a more general persuasion platform problem:

\begin{proposition}
    Given a generalized persuasion platform problem $(\Omega,\Theta,\mu,x^*,\{u_R^\theta\}_{\theta \in \Theta})$, consider the equivalent market segmentation problem $(V,\pi^*)$. Assume that the following conditions hold:
    \begin{enumerate}
        \item For any user distribution $x \in \Delta(\Theta)$, there exists $\theta \in \Theta$ such that $p_\theta \in \arg\max_{p} U_S(p;x)$. 
        \item There exists a constant $\alpha > 0$ such that for any $\theta, \theta' \in \Theta$, $S_{\theta} - S_{\theta'} = \alpha R_{\theta \to \theta'}$.
        \item For any $\theta, \theta' \in \Theta$, if $ R_{\theta \to \theta'} < 0$ then playing $a_0$ regardless of the sender's recommendation is a best-reply of the user of type $\theta$ against the sender's strategy $p_{\theta'}$.
    \end{enumerate}
    Then, any consumer surplus maximizing segmentation $\sigma \in \Sigma$ in the market segmentation problem is an optimal platform policy in the generalized persuasion platform problem.
\end{proposition}

\begin{proof}
    Similarly to the case of binary state space, it is sufficient to show that the games are strategically equivalent (see Theorem \ref{thm1} and Corollary \ref{cor3}).
    Let $x \in \Delta(\Theta)$ be a user distribution, and let $p_{\theta'}$ be an optimal sender strategy (from the \emph{first condition}, there exists an optimal sender strategy of this form).
    For the posterior user distribution $x$ (and its corresponding market $\pi$), the utility of the sender in the persuasion problem is $S_{\theta'}$ times the mass of users who follow its recommendation (from the \emph{third condition}, these are only the users $\theta$ for which $R_{\theta \to \theta'} \ge 0$), and the utility for the consumers in the segmentation problem is the corresponding optimal price $S_{\theta'}$ times the mass of consumers for which $v_\theta - v_{\theta'} = S_\theta - S_{\theta'} \ge 0$. From the \emph{second condition}, it holds that the two terms are equal, hence the sender and producer utilities are the same in the two problems.

    It is now left to show that the users and consumers have the same incentive. Note that indeed the \emph{second condition} guarantees that a user of type $\theta$ will opt-in when the sender's posterior is $x$ if and only if a consumer whose valuation is $v_{\theta}$ will purchase the product in the corresponding segment $\pi$. 
    Moreover, since the constant $\alpha$ is uniform across all pairs of types, it follows that the consumers' surplus in $\pi$ equals the users' utility in $x$ up to a positive multiplicative factor $\alpha$, which is independent of the distribution $x$. Therefore, a segmentation that maximizes the average consumer surplus in the market segmentation problem also maximizes the users' average utility in the persuasion platforms problem.
\end{proof}




Note that in the standard binary case, the first condition holds trivially, and the second condition holds with the constant $\alpha = 1$. As for the third condition, note that it holds in the following two general cases:
\begin{enumerate}
    \item For any $\omega \in \Omega$, the sign of $u_R^{\theta}(a_1,\omega)$ does not depend on $\theta$.
    \item There exists a complete order on $\Omega$, and for any $\theta$, $u_R^{\theta}(a_1,\omega)$ is monotone in $\omega$.
\end{enumerate}
The standard binary state model falls into the first case, and the two examples that follow fall into these two cases as well.

We now turn to provide two concrete, natural examples for state spaces and utility schemes for the users, for which the above conditions holds, implying the persuasion problem can be solved simply by reducing to the equivalent market segmentation problem.

\paragraph{Multiple states with a single global low-quality state} Consider a finite state space $\Omega=\{\omega_0,\omega_1, \ldots, \omega_k\}$ of cardinality $k+1$, in which at states $\omega_1, ..., \omega_k$ playing $a_1$ yields some positive utility for all user types (depending on their specific utility function), and the "bad" state $\omega_0$ yields a constant negative utility of $-1$ whenever there is a trade. The motivation is as follows: the seller can provide either one of $k$ copies of the product which are of some decent quality (which is subjective and depends on the user type), or she can deliver some broken product that is equally useless for all user types. Formally, a user type $\theta \in \Theta$ is now a vector is $\R_{+}^k$, and the utility function of a $\theta-$type user is now given by:
\begin{equation*}
    u^{\theta}_R(a,\omega) = \begin{cases}
        0  & a = a_0 \\
        -1 & a = a_1 \text{ and }\omega = \omega_0 \\
        \theta_l & a = a_1 \text{ and }\omega = \omega_l\text{ for } l=1, ... , k
\end{cases}
\end{equation*}

Note that the standard binary state case presented in our model is a private case of this multiple states case, with $k=1$. We further denote $\mu_l = \mu(\omega_l)$. It is straightforward to show that a sender's optimal policy against a user of a known type $\theta$ is to recommend $\tilde{a}_1$ with probability $1$ whenever the state is other than $\omega_0$, and when the state is $\omega_0$ recommend $\tilde{a}_1$ with the following probability (which makes the user indifferent):

\begin{equation*}
    \tau_\theta = \frac{1}{\mu_0} \sum_{l=1}^k \mu_l \theta_l
\end{equation*}

Hence, the sender's utility against a type $\theta$ user when playing optimally is given by $S_\theta = (1 - \mu_0) + \mu_0 \tau_\theta$, and $R_{\theta \to \theta'}$ can now be written as:

\begin{equation*}
    R_{\theta \to \theta'} = \sum_{l=1}^k \mu_l \theta_l + \mu_0 \tau_{\theta'} = \mu_0 (\tau_{\theta} - \tau_{\theta'})
\end{equation*}

implying the second condition holds with $\alpha = 1$. Lastly, the first condition holds trivially, and the third condition holds since for any $\omega \in \Omega$, the sign of $u_R^{\theta}(a_1,\omega)$ does not depend on $\theta$.

\paragraph{Continuous state space} Consider a normalized continuum of states, namely $\Omega = [0,1]$, equipped with some nonatomic distribution $\mu \in \Delta(\Omega)$. The user type is now identified with a single parameter $\theta\in[0,1]$ which stands for the threshold quality of the product from which it considers the product as valuable. Formally, the utility of a $\theta-$type user from buying is given by $u^{\theta}_R(a_1,\omega) = \ind{\omega \ge \theta} - c$ for some constant $c \in (0,1)$, which stands for the cost (e.g., price) of the product, which is the same for all user types. The utility of not buying the product remains zero.

As for the sender's optimal strategy against a given user type $\theta$, it can be shown that there always exists a deterministic threshold policy, in which the sender recommends $\tilde{a}_1$ if and only if the true state $\omega$ is at least some $f(\theta)$, where $f(\theta) < \theta$.\footnote{While it is clear that $f(\theta)$ also depends on the cost $c$, we omit it from the notation for brevity (as the cost is uniform across all user types).} The threshold of the persuasion policy is again determined using the indifferent condition of the user. The sender's utility from targeting a $\theta-$type user is then simply $S_\theta = \mu(\omega \ge f(\theta))$. As for $R_{\theta \to \theta'}$,

\begin{align*}
    R_{\theta \to \theta'} &= \E_{\omega \sim \mu} [\ind{\omega \ge \theta} - c | \omega \ge f(\theta')] \cdot \mu(\omega \ge f(\theta')) \\
    &= \mu(\omega \ge \theta | \omega \ge f(\theta')) \mu(\omega \ge f(\theta')) - c \mu(\omega \ge f(\theta')) \\
    &= \mu(\omega \ge \theta) - c \mu(\omega \ge f(\theta'))
\end{align*}

Now, since $R_{\theta \to \theta} = 0$ by the construction of the sender's optimal strategy, subtracting it does not change the term, hence:

\begin{align*}
    R_{\theta \to \theta'} &= R_{\theta \to \theta'} - R_{\theta \to \theta} \\
    &= \mu(\omega \ge \theta) - c \mu(\omega \ge f(\theta')) - (\mu(\omega \ge \theta) - c \mu(\omega \ge f(\theta))) \\
    &= c(\mu(\omega \ge f(\theta)) - \mu(\omega \ge f(\theta'))) = c(S_{\theta} - S_{\theta'})
\end{align*}

implying the second condition holds with $\alpha = \frac{1}{c}$. Again the first condition holds trivially, and the third condition holds since for any $\theta$, $u_R^{\theta}(a_1,\omega)$ is monotone in $\omega$.

\section{Reputation-based Persuasion Platforms}
\label{sec:reputation-based-pp}

\subsection{The Model}
\label{subsec:rbpp-model}

We now turn to extend the persuasion platform model to a repeated case, in which myopic users, drawn i.i.d from the prior user distribution $x^*\in\Delta(\Theta)$, arrive sequentially for an infinite number of time periods. In this setting the platform has the ability to impose an irreversible punishment on the sender. The punishment results in a loss for both the sender and the platform.
The setting may be motivated by reputation considerations where the platform and the sender are engaged in a contract in which the platform requires the sender to fully reveal the state on a subset of signals. If the sender violates the contract it is thrown out of the platform to some less desirable outside option.

Formally, on the outset, the platform commits to an information revelation policy $\sigma: \Theta \to \Delta(S)$ which will be used for all subsequent time periods, and additionally it specifies a subset of signals $S_T \subset S$ (or, equivalently, on a subset of the posteriors $X_T \subset \operatorname{supp}(\sigma)$), on which it requires from the sender to act truthfully. i.e., recommend a product of low quality with probability zero.\footnote{We restrict attention to stationary platform policies, meaning that both $\sigma$ and $S_T$ are time-independent. While this may not be user-welfare maximizing, one possible justification is as follows: a non-stationary platform policy is unfair for the users, in the sense that the expected utility of a myopic user of a certain type may depend on its arrival time – meaning two identical users may expect different utilities just because they arrived on different times. The fact that we restrict attention to stationary policies prevents this from happening.}

The sender then commits to its own persuasion policy $p: S \to \Delta(A)$ (where again we identify $p(s)$ with the probability of recommending a low-quality product after observing $s \in S$), and then users begin to arrive one by one.

The platform maintains a reputation for the sender, which can be either high or low at the beginning of each time period. We assume that at the beginning of the first time period, the sender has a high reputation.

If at time period $t$ the reputation of the sender is high, then the interaction between the platform, the sender, and the user at time $t$ is the same as in the one-shot model, with the distinction that now the platform also observes the satisfaction level of the user at the end of the interaction. We assume that the satisfaction level of a user from the interaction is bad if it purchased a product of low quality, and good otherwise.\footnote{Our model can be extended to the case where the feedback on user satisfaction is noisy, i.e., it is correct only with some known probability $q$. In this case, the analysis can be adapted by plugging this probability into the transition probabilities between reputation states.}

If the platform provided to the sender a signal on which it requires truthful behavior, and the satisfaction level the platform observes is bad (namely, the sender had a high reputation, and manipulated the user into buying a product of low quality), then the sender is being permanently moved to the low reputation state for all future time periods. Otherwise, the sender begins the next round with a high reputation. Formally, the sender is being moved from high to low reputation after time period $t$, if and only if $s_t \in S_T$, $\omega_t = \omega_0$ and $a_t = a_1$.

At the low reputation state, the sender suffers from a fixed punishment utility $\bar{u}$, satisfying $\bar{u} < U_S(p^*;x^*)$. That is, in the low reputation state the sender's utility is strictly lower than in the case where it has no user information at all. One can interpret the low reputation state as an alternative market or an external option, that is less favorable to the sender relative to the scenario where the platform does not provide any information about the user's preferences. When at low reputation, the utility for the user is defined to be zero.

The sender then plays to optimize its discounted utility stream for some discount factor $0 < \delta < 1$, while the platform optimizes the non-discounted average user utility.

Let $\Sigma$ be the set of all Bayes-plausible distributions with respect to the prior distribution $x^*$. A policy of the platform is now a tuple $(\sigma, X_T)$, where $\sigma \in \Sigma$ and $X_T \subset \operatorname{supp}(\sigma)$ is the subset of posterior on which the platform requires the sender to play truthfully. Note that once fixing a platform policy $(\sigma, X_T)$, one can think of the sender as an agent operating in an induced Markov decision process (MDP), as demonstrated in Figure \ref{sender-mdp}.

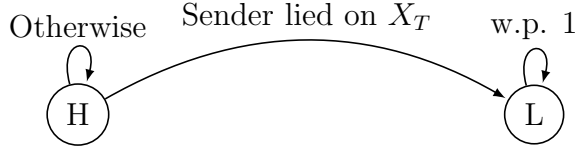
\begin{figure}
    \centering
    \begin{tikzpicture}[->,>=latex,shorten >=1pt,auto,node distance=6cm,semithick]
      \tikzset{state/.style={circle,fill=white,draw=black,text=black}}
    
      \node[state] (H) {H};
      \node[state] (L) [right of=H] {L};
    
      \path (H) edge [loop above] node {Otherwise} (H)
            (H) edge [bend left] node {Sender lied on $X_T$} (L)
            (L) edge [loop above] node {w.p. 1} (L);
    \end{tikzpicture}
    \caption{Reputation-based Persuasion Platform as an MDP for the sender.}
    \label{sender-mdp}
\end{figure}

Unlike the one-shot problem, the reputation-based persuasion problem cannot be reduced to an equivalent market segmentation problem with reputation. The reason is that in the one-shot identification, we identified prices with the probability of recommending a low-quality product. In the repeated market segmentation problem, prices are verifiable and this can be used by a platform to deter the sender. In contrast, in our case, the only verifiable information is whether a low-quality product has been purchased. This implies that the strategic problems in the repeated setting are no longer equivalent.  

Without loss of generality, we assume that $|X_T|=1$, and denote by $x^T$ the single posterior in $X_T$.\footnote{This is due to the fact that for any incentive-compatible policy satisfying $|X_T|>1$, one can construct a new incentive-compatible policy providing the same average user utility, for which $|X_T|=1$ by simply merging all of the posteriors in $X_T$.} To shorten, we refer to the tuple $(\sigma, x_T)$ by $\sigma$ only. We further denote:

\begin{equation}
\begin{gathered}
    \operatorname{supp}(\sigma) = \set{x^1, ... x^m, x^T } \nonumber \\
    \forall i \in [m]: \alpha^i = \sigma(x^i) \nonumber \\
    \alpha^T = \sigma(x^T) \nonumber 
\end{gathered}
\end{equation}

Denoting $p_j = \tau_{\theta_j}$, $I_j(x,p)=x_j \ind{p(x) \le p_j}$, the sender and platform one-shot utilities at the high reputation state are given by,

\begin{align}
U_S(\sigma,p) &= \alpha^T \sum_{j=1}^n I_j(x^T,p)(\mu + (1-\mu)p(x^T)) \nonumber \\
&\quad+ \sum_{i=1}^m \alpha^i \sum_{j=1}^n I_j(x^i,p)(\mu + (1-\mu)p(x^i)) \nonumber \\
\nonumber
\end{align}

\begin{align}
U_P(\sigma,p) &= \alpha^T \sum_{j=1}^n I_j(x^T,p)(\mu \theta_j - (1-\mu)p(x^T)) \nonumber \\
&\quad+ \sum_{i=1}^m \alpha^i \sum_{j=1}^n I_j(x^i,p)(\mu \theta_j - (1-\mu)p(x^i)) \nonumber
\end{align}

We now define two types of possible sender strategies:
\begin{definition}
A strategy of the sender $p^*$ is greedy if for all $x \in \Delta(\Theta)$: 
\begin{equation}
    p^*(x) = \argmax_p U_S(p;x) \nonumber
\end{equation}
\end{definition}

\begin{definition}
A strategy of the sender $p_T$ is truthful (with respect to a given platform policy $\sigma$) if there exists a greedy strategy $p^*$ such that for all $x \in \Delta(\Theta)$: 
\begin{equation}
    p_T(x) = \begin{cases}
0 &\text{$x = x^T$}\\
p^*(x) &\text{$x \neq x^T$} \nonumber \\
\end{cases}
\end{equation}
\end{definition}

We assume that when the sender is indifferent between targeting multiple types of users (i.e., when $\argmax_p U_S(p;x)$ is not uniquely defined), it targets the lowest type among them. That is, we uniquely define the optimal sender strategy to be $p^*(x) = \min \argmax_p U_S(p;x)$. Therefore, the truthful strategy $p_T$ is also uniquely defined.

Note that without loss of generality, we can only consider platform policies that are incentive-compatible, i.e., policies in which the sender does not benefit from deviating from being truthful (lying with zero probability) when it operates at the posterior $x^T$ (clearly, the sender will always be greedy at any other posterior, as lying at such posterior has no consequences at all):

\begin{definition}
    A policy of the platform $\sigma \in \Sigma$ is incentive-compatible (IC) if $p_T$ is the sender's best response with respect to $\sigma$.
\end{definition}

For any $x \in \Delta(\Theta)$ and $k \in [n]$, we denote by $F_k(x) = \sum_{j=k}^n x_j$ the mass of users in $x$ whose persuasion threshold is weakly above $p_k$. Given a policy $\sigma$, we denote by $\sigma_F = \sigma |_{x \neq x^T}$ the policy obtained by conditioning on the posterior to be any non-truthful posterior, and denote by $x^F = \E_{x \sim \sigma_F} [x] = \E_{x \sim \sigma} [x | x \neq x^T]$ its mean. Note that one could decompose the sender and platform utilities as follows:

\begin{equation}
\begin{gathered}
    U_S(\sigma,p)=\alpha^T U_S(p;x^T) + (1-\alpha^T) U_S(\sigma_F,p) \nonumber \\
    U_P(\sigma,p)=\alpha^T U_P(p;x^T) + (1-\alpha^T) U_P(\sigma_F,p) \nonumber
\end{gathered}
\end{equation}

We denote by $V(\sigma) = U_S(\sigma, p_T)$ the sender's one-shot utility when playing truthfully with respect to a platform policy $\sigma$. Since playing $p_T$ ensures that the sender remains in the high reputation state, $V(\sigma)$ is also the overall, long-term utility of the sender when playing truthfully. In particular, $V(\sigma)$ is the sender's long-term utility when responding optimally to an incentive-compatible platform policy $\sigma$.

Notice that in the reputation-based setting, the problem has a straightforward solution in the following cases: first, for a given fixed punishment level $\bar{u}$, and for large enough discount factor $\delta$, the policy where the platform requests from the sender a completely truthful behavior on all signals is incentive-compatible and therefore optimal. This follows from the fact that the sender's utility from a deviation approaches $\bar{u}$ as $\delta$ goes to one. Similarly, for a fixed discount factor $\delta$ and for a low enough punishment level $\bar{u}$, the platform can also require a truthful behavior of the sender, which again will be incentive-compatible. We aim to study the platform's optimal signaling policy in the general case, where $\bar{u}$ and $\delta$ are such that the sender might benefit from deviation from the truthful strategy with respect to the platform's request.

\subsection{Optimal Platform Policy Characterization}
\label{subsec:rbpp-optimal-platform}

In this section, we provide a useful characterization of an optimal platform policy in the reputation-based setting, which will then be used to simplify the platform's optimization problem of finding an optimal policy. We start by defining two standard properties of a platform policy:

\begin{definition}
    A platform policy $\sigma \in \Sigma$ is Pareto-efficient if $\forall \tilde{\sigma} \in \Sigma$:
    \begin{enumerate}
        \item $U_S(\sigma, p_T) < U_S(\tilde{\sigma}, p_T) \Rightarrow U_P(\sigma, p_T) > U_P(\tilde{\sigma}, p_T)$
        \item $U_P(\sigma, p_T) < U_P(\tilde{\sigma}, p_T) \Rightarrow U_S(\sigma, p_T) > U_S(\tilde{\sigma}, p_T)$
    \end{enumerate}
\end{definition}

\begin{definition}
    A platform policy $\sigma \in \Sigma$ is lowest-type-targeting if $\forall x \in \operatorname{supp}(\sigma_F)$:
    \begin{equation}
        p^*(x) = \frac{\mu}{1-\mu} \theta^x_{min} \nonumber
    \end{equation}
    where $\theta^x_{min} = \min \operatorname{supp}(x)$.
\end{definition}

That is, a platform policy is lowest-type-targeting if, for every user distribution in its support, the sender's optimal policy is to play according to the lowest type's persuasion threshold (i.e., it does not benefit from increasing the probability of lying at the price of losing the lowest type users). Notice that when a platform policy is lowest-type-targeting, all users always follow the sender's recommendation, hence all products are sold. In Lemma \ref{lem9} we show that this property is equivalent to Pareto-efficiency.

We start with a standard result showing that the set of all optimal platform policies in the reputation-based persuasion platform setting must contain at least one Pareto-efficient policy $\sigma$, for which $U_S(\sigma_F, p^*) = U_S(p^*; x^F)$ (that is, the property of \cite{bergemann2015} is satisfied with respect to the policy conditioned on the fact that the platform does not require truthful behavior). Note that we define an optimal platform policy as an incentive-compatible policy maximizing the platform utility (which is the average user utility) over all incentive-compatible platform policies. That is, denoting by $\Sigma^*$ the set of all optimal platform policies, we have that $\Sigma^* \subseteq \Sigma_{IC}$ where $\Sigma_{IC}$ is the set of all incentive-compatible platform policies. The optimization problem of the platform can be then written as follows:

\begin{equation}
\label{opt}
\begin{aligned}
\max_{\sigma \in \Sigma_{IC}} \quad & U_P(\sigma, p_T)
\end{aligned}
\end{equation}

Note that the problem defined in \eqref{opt} is an infinite-dimensional optimization problem. To simplify the optimization problem, we introduce an optimal platform policy characterization in Theorem \ref{thm3}:

\begin{theorem}
    \label{thm3}
    Denote by $\Sigma^*$ the set of all optimal platform policies. Then, there exists a Pareto-efficient policy $\sigma \in \Sigma^*$ such that $U_S(\sigma_F, p^*) = U_S(p^*; x^F)$.
\end{theorem}

That is, there exists an optimal platform policy for the reputation-based setting, such that conditional on the fact that the platform does not require truthful behavior from the sender, its signaling policy is the same as in the one-shot case, and therefore can be obtained by applying the algorithm of \cite{bergemann2015}. Theorem \ref{thm3} is then used to simplify the optimization problem of finding the optimal platform policy as follows:

\begin{corollary}
\label{cor4}
  The platform's optimization problem can be solved using the following two-step procedure:
  \begin{itemize}
    \item Solve the following finite-dimensional optimization problem: 
    \begin{equation}
    \label{opt-simple}
    \begin{aligned}
    \min_{\sigma=(\alpha, x^T, x^F)} \quad & U_S(\sigma, p_T)\\
    \textrm{s.t.} \quad 
    & 0 \le \alpha \le 1 \\
    & x^T, x^F \in \Delta(\Theta) \\
    & \alpha x^T + (1-\alpha) x^F = x^*\\
    & \sigma \in \Sigma_{IC} \\
    \end{aligned}
    \end{equation}
  \item Apply the algorithm of \cite{bergemann2015} with respect to $x^F$.
  \end{itemize}
\end{corollary}

Corollary \ref{cor4} follows immediately from Theorem \ref{thm3}: the first stage yields the lowest sender utility that can be achieved by an incentive-compatible platform policy (follows from Theorem \ref{thm3}). Hence, the maximal platform utility is bounded from above by the point on the Pareto frontier in which the Sender achieves the above minimal utility. In order to complete the construction we need to show that this platform utility is indeed achievable.
To see this, we use the same approach as in the one-shot case: we apply the reduction to \cite{bergemann2015}
 with respect to  $x^F$, and achieve the Pareto-efficient utility which is maximal for the platform and does not effect the sender's utility. Note that incentive compatibility does not break since the utility for the sender and $x^T$ are not being modified (follows from Corollary \ref{cor5} below). 

Unlike the problem defined in \eqref{opt}, the simplified problem in \eqref{opt-simple} is a finite-dimensional optimization problem, and is more likely to be solved analytically or numerically. However, the optimization problem is non-convex, and therefore it is still generally difficult (this is due to the fact that the incentive-compatibility constraints are non-convex, see Lemma \ref{lem2}).

\subsection{Proof of Theorem \ref{thm3}}
\label{subsec:rbpp-proof}

To show this characterization we first introduce some auxiliary lemmas which will be used in the proof of Theorem \ref{thm3} (and the proof of Theorem \ref{thm3} will then follow). To begin with, notice that for a given platform policy $\sigma$, $p_T(x)=p^*(x)$ for any posterior $x \neq x^T$. Moreover, it is clear that playing greedily at any non-truthful posterior is optimal for the sender, as it maximizes its one-shot utility and cannot lead to punishment (i.e., the sender being moved by the platform from a high to a low reputation state). Therefore, any potentially profitable deviation from $p_T$ must be conditioned on the sender being at the truthful signal $x^T$. This observation leads to the following alternative definition of incentive compatibility:

\begin{lemma}
\label{lem2}
    A policy of the platform $\sigma \in \Sigma$ is incentive-compatible if and only if for all $k \in [n]:$ 
    \begin{align}
            V(\sigma) \ge \frac{1-\delta}{\delta} \cdot \left(1 - \frac{1 - F_k(x^T)}{F_k(x^T) \theta_k}\right) + \bar{u} \nonumber
    \end{align}
\end{lemma}

A key corollary regarding incentive-compatibility is that if two policies $\sigma, \tilde{\sigma} \in \Sigma$ have the same truthful distribution and guarantee the same sender utility from being truthful, then either both are incentive-compatible or both are not:

\begin{corollary}
\label{cor5}
    Let $\sigma, \tilde{\sigma} \in \Sigma$ such that $x^T = \tilde{x}^T$ and $V(\sigma) = V(\tilde{\sigma})$. Then $\sigma$ is incentive-compatible if and only if $\tilde{\sigma}$ is incentive-compatible.
\end{corollary}

We continue with the following three technical lemmas, implying that given a platform policy $\sigma \in \Sigma$ such that $U_S(\sigma_F, p^*) > U_S(p^*; x^F)$, there exists policies $\sigma_l, \sigma_h \in \Sigma$ such that $V(\sigma_l) < V(\sigma) < V(\sigma_h)$:

\begin{lemma}
\label{lem3}
    Let $\epsilon > 0$. Given a platform policy $\sigma$, define $\Sigma_\epsilon = \{ \tilde{\sigma} \in \Sigma: \tilde{\alpha}^T = (1-\epsilon) \alpha^T \wedge \tilde{x}^T = x^T \}$. Then for any $\sigma_\epsilon \in \Sigma_\epsilon$, $x^F_\epsilon \to x^F$ as $\epsilon \to 0$.
\end{lemma}

\begin{lemma}
\label{lem4}
    Let $\sigma \in \Sigma$ such that $U_S(\sigma_F, p^*) > U_S(p^*; x^F)$. Then, there exists some $0 < \epsilon^* < 1$ and $\sigma_l \in \Sigma_{\epsilon^*}$ such that $V(\sigma_l) < V(\sigma)$.
\end{lemma}

\begin{lemma}
\label{lem5}
    Let $\sigma \in \Sigma$ such that $U_S(\sigma_F, p^*) > U_S(p^*; x^F)$. Then, for any $0 < \epsilon < 1$, there exists $\sigma_h \in \Sigma_{\epsilon}$ such that $V(\sigma_h) > V(\sigma)$.
\end{lemma}

We now turn to introduce a set of lemmas dealing with the connections between Pareto-efficiency, lowest-type-targeting, and optimality of platform policies:

\begin{lemma}
\label{lem6}
    For any lowest-type-targeting $\sigma \in \Sigma$:
    \begin{enumerate}
        \item $U_S(\sigma_F, p^*) = \mu + \frac{1-\mu}{1-\alpha^T} \sum_{i=1}^m \alpha^i p^*(x^i)$
        \item $V(\sigma) = U_S(\sigma, p_T) = \mu + (1-\mu) \sum_{i=1}^m \alpha^i p^*(x^i)$
        \item $U_P(\sigma_F, p^*) = \frac{\mu}{1-\alpha^T} \sum_{i=1}^m\alpha^i (\sum_{j=1}^n x^i_j \theta_j) - \frac{1-\mu}{1-\alpha^T} \sum_{i=1}^m \alpha^i p^*(x^i)$
        \item $U_P(\sigma, p_T) = \mu \sum_{j=1}^n x^*_j \theta_j - (1-\mu) \sum_{i=1}^m \alpha^i p^*(x^i)$
    \end{enumerate}
\end{lemma}

\begin{lemma}
\label{lem7}
    For any platform policy $\sigma \in \Sigma$ which is not lowest-type-targeting, there exists a platform policy $\tilde{\sigma} \in \Sigma$ such that $x^T = \tilde{x}^T$, $\alpha^T = \tilde{\alpha}^T$, $U_P(\tilde{\sigma},p_T) \ge U_P(\sigma,p_T)$ and $U_S(\tilde{\sigma},p_T) > U_S(\sigma,p_T)$.
\end{lemma}

\begin{lemma}
\label{lem8}
    There exists an optimal policy $\sigma \in \Sigma^*$ which is lowest-type-targeting.
\end{lemma}

\begin{lemma}
    \label{lem9}
    A platform policy $\sigma \in \Sigma$ is lowest-type-targeting if and only if it is Pareto-efficient.
\end{lemma}

The following corollary then follows directly from Lemma \ref{lem8} and Lemma \ref{lem9}:

\begin{corollary}
\label{cor6}
    Denote by $\Sigma_P$ the set of all Pareto-efficient platform policies. Then, the set of all Pareto-efficient optimal platform policies $\Sigma^* \cap \Sigma_P$ is nonempty.
\end{corollary}

We now turn to prove Theorem \ref{thm3}, which provides a characterization of the optimal platform policy:

\begin{proof}[\textbf{Proof of Theorem \ref{thm3}.}]
Assume by contradiction that this is not the case. That is, for any $\sigma \in \Sigma^*$ we have $U_S(\sigma_F, p^*) > U_S(p^*; x^F)$. Denote by $\Sigma_{P}$ the set of all Pareto-efficient platform policies. From Corollary \ref{cor6}, the set $\tilde{\Sigma} \coloneqq \Sigma^* \cap \Sigma_P$ is not empty.

Let $\sigma \in \argmin_{\sigma' \in \tilde{\Sigma}} {\alpha'^T}$. From the assumption we know that $U_S(\sigma_F, p^*) > U_S(p^*; x^F)$. We will show that there exists $\hat{\sigma} \in \tilde{\Sigma}$ such that $\hat{\alpha}^T < \alpha^T$, in contradiction to the minimality of $\alpha^T$.

For any $0 < \epsilon < 1$, let us define the following set of platform policies: 
\begin{align*}
    \Sigma_\epsilon = \set[\Big]{\tilde{\sigma} \in \Sigma: \tilde{\alpha}^T = (1-\epsilon) \alpha^T \wedge \tilde{x}^T = x^T}
\end{align*}

and note that for every $\sigma_\epsilon \in \Sigma_\epsilon$, $x^F_\epsilon = \E_{x \sim \sigma^F_\epsilon} [x] = \E_{x \sim \sigma_\epsilon} [x | x \neq x^T]$ is uniquely defined. From Lemma \ref{lem4} there exist $0 < \epsilon^* < 1$ and $\sigma_l \in \Sigma_{\epsilon^*}$ such that $V(\sigma_l) < V(\sigma)$. From Lemma \ref{lem5}, for the same $\epsilon^*$ there exists $\sigma_h \in \Sigma_{\epsilon^*}$ such that $V(\sigma_h) > V(\sigma)$. $V$ is a continuous function and $\Sigma_{\epsilon^*}$ is a convex set, therefore one can take a convex combination of the two policies and construct a new policy $\tilde{\sigma} \in \Sigma_{\epsilon^*}$ such that $V(\tilde{\sigma})=V(\sigma)$. From Corollary \ref{cor5}, $\tilde{\sigma}$ is incentive-compatible (as $\sigma$ is incentive compatible, $\tilde{x}^T = x^T$ and $V(\tilde{\sigma})=V(\sigma)$). It is now left to show that $\tilde{\sigma}$ provides the same average user utility as $\sigma$ does.

First, we recall that \cite{bergemann2015} showed that for every segmentation $\tilde{\sigma}_F \in \Sigma$ there exists a Pareto-efficient segmentation $\hat{\sigma}_F \in \Sigma$ such that $W_S(\tilde{\sigma}_F, \phi^*) = W_S(\hat{\sigma}_F, \phi^*)$. Using their claim it can be shown that there exists a platform policy $\hat{\sigma} \in \Sigma_{\epsilon^*}$ such that $U_S(\tilde{\sigma}_F, p^*) = U_S(\hat{\sigma}_F, p^*)$, and $\hat{\sigma}$ is Pareto-efficient. Note that $\hat{\sigma}$ is still incentive-compatible from the same arguments concerning $\tilde{\sigma}$.

Using the fact that both $\sigma$ and $\hat{\sigma}$ are Pareto-efficient, we get that:
\begin{align*}
    V(\hat{\sigma}) = V(\sigma) \Leftrightarrow U_P(\hat{\sigma}, p_T) = U_P(\sigma, p_T)
\end{align*}

therefore, $\hat{\sigma} \in \tilde{\Sigma}$ and $\hat{\alpha}^T = (1-\epsilon^*)\alpha^T < \alpha^T$, in contradiction to the fact that $\sigma \in \argmin_{\sigma' \in \tilde{\Sigma}} {\alpha'^T}$.

\end{proof}

\subsection{A Numerical Example: Binary User Type}
\label{subsec:rbpp-numerical}

In this section, we provide an example of how Corollary \ref{cor4} enables a numerical solution for the reputation-based persuasion platforms problem, and yields several economic insights and implications. The analysis focuses on the first stage of the two-stage procedure, in which the platform splits the prior user distribution into two posteriors $x^T$ (in which it requires truthful behavior) and $x^F$ (in which the sender is allowed to play greedily).
We first note that even in the binary user type case (namely $n=2$), the problem is hard to solve analytically. This is due to the fact that the problem is a nested optimization problem (as the objective is a maximum function itself), and the search is performed over a non-convex set, defined by both Bayes-plausibility constraints and performance constraints in the form of lower bounds on the objective itself (where the bounds depend on the decision variables themselves).

However, we note that in low-dimension, an approximation of an optimal solution can be obtained using an exhaustive search over the simplex. We utilize this fact to approximate the sender's utility, as well as the distribution of users on which it acts truthfully. Notice that both of these don't change during the second stage of the two-stage procedure, hence the first stage is sufficient.

For this example, we fix $\mu=\frac{1}{2}, \theta_1 = \frac{1}{10}, \theta_2 = \frac{7}{10}$ and $\bar{u}=0.1$ (which is indeed lower than the monopolistic utility of the sender). We consider three different cases of the prior user distribution: harder-to-persuade users $\theta_1$ are less likely (30\%), equality likely (50\%), and more likely (70\%). For each of these prior distributions, we study the effect of the sender's discount factor $\delta$ on both the sender utility and the likelihood of a user being a harder-to-persuade user conditional on truth-requiring of the platform (i.e., $x_{1}^T$). We treat the structure of the truthful posterior as a measure of \emph{fairness}, as it measures which type of user benefits the most from the platform's truth-requiring strategy.
For every prior distribution, we restrict attention to the range of discount factors in which the solution is non-trivial. The results are summarized in Figure \ref{fig:experiment-meta-figure}.\footnote{Alternatively, for a fixed discount factor $\delta$, decreasing $\bar{u}$ has a similar effect as increasing $\delta$ for a fixed $\bar{u}$, as both operations translate into an increase in the platform punishment capabilities.}

\begin{figure}[htp]
    \centering
    \subfloat[Sender's utility as a function of $\delta$ for $x^*=(0.3,0.7)$]{\includegraphics[width=0.4\textwidth]{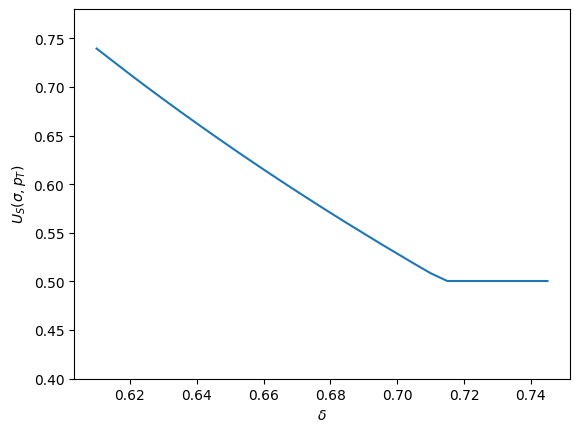}\label{fig:u30}}
    \hfill
    \subfloat[Posterior $x_1^T$ as a function of $\delta$ for $x^*=(0.3,0.7)$]{\includegraphics[width=0.4\textwidth]{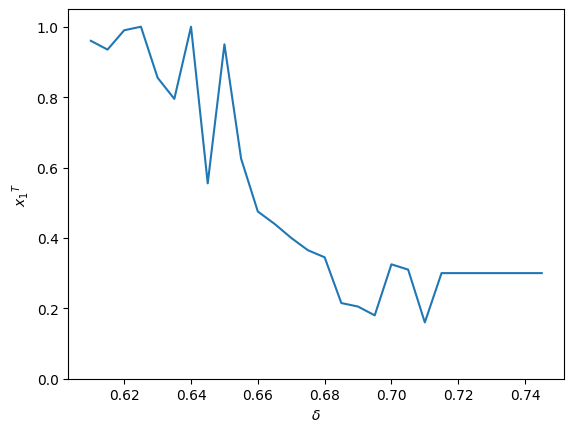}\label{fig:x30}}
    
    \subfloat[Sender's utility as a function of $\delta$ for $x^*=(0.5,0.5)$]{\includegraphics[width=0.4\textwidth]{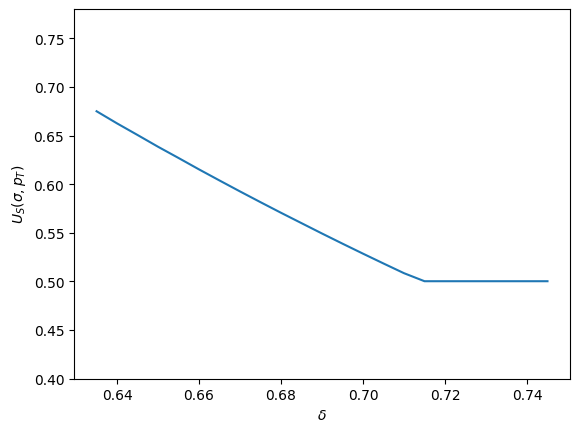}\label{fig:u50}}
    \hfill
    \subfloat[Posterior $x_1^T$ as a function of $\delta$ for $x^*=(0.5,0.5)$]{\includegraphics[width=0.4\textwidth]{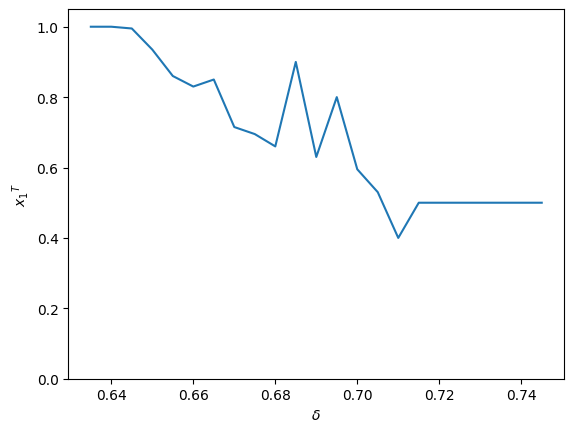}\label{fig:x50}}
    
    \subfloat[Sender's utility as a function of $\delta$ for $x^*=(0.7,0.3)$]{\includegraphics[width=0.4\textwidth]{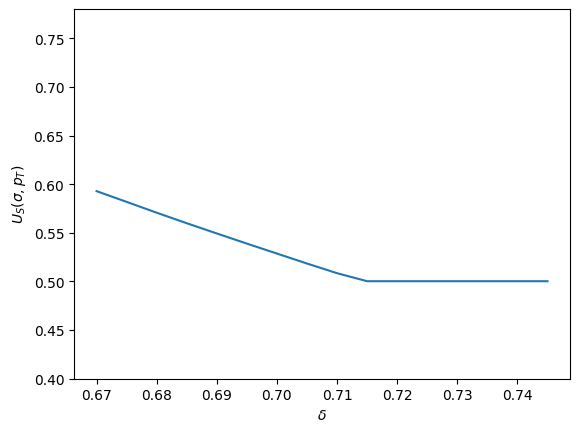}\label{fig:u70}}
    \hfill
    \subfloat[Posterior $x_1^T$ as a function of $\delta$ for $x^*=(0.7,0.3)$]{\includegraphics[width=0.4\textwidth]{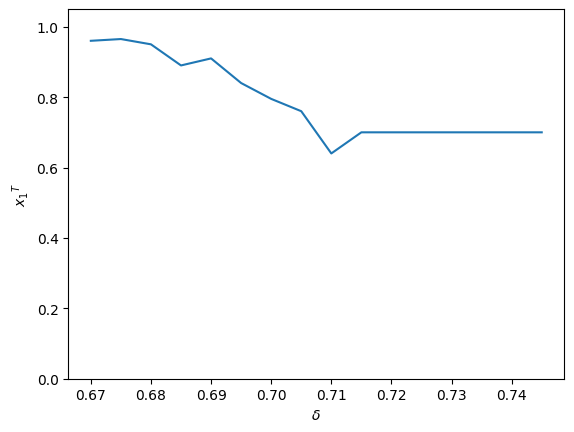}\label{fig:x70}}
    
    \caption{Utility of the sender and truthful posterior as functions of $\delta$ for different priors. Each row represents a different prior (30\%, 50\%, and 70\%), with the left column showing the sender's utility and the right column displaying the truthful posterior.}
    \label{fig:experiment-meta-figure}
\end{figure}

\paragraph{Utility of the sender}
It is notable that for all prior user distributions, the sender's utility decreases as the discount factor increases. 
This outcome is expected: when the sender is more careful about future timesteps, the platform gains more power to penalize, using this ability to lower the sender's utility subject to incentive-compatibility.
For a large enough $\delta$ it becomes incentive-compatible for the platform to require truthful behavior over the entire posterior with probability 1, leading to a profit of $\mu = \frac{1}{2}$ for the sender (as it only sells the product when its quality is high).

\paragraph{Posterior on which the sender acts truthfully}  When $\delta$ is large enough and complete truth-requiring becomes incentive-compatible, it is clear that $x^T = x^*$. Interestingly, as $\delta$ decreases, it appears that the distribution $x^T$ assigns more mass to the harder-to-persuade users, leading to a more biased (hence less fair) distribution of users benefiting from the platform's optimal strategy.
While this monotonicity trend holds only globally rather than locally, we highlight that fluctuations are caused due to the use of the exhaustive search method (which only finds an approximate solution). 

\section{Discussion}
\label{sec:discussion}

In this paper, we introduced a bi-level Bayesian persuasion model in which a third-party platform controls the information available to the sender about user preferences. We characterized the optimal information disclosure policy of the platform, which maximizes average user utility, in a subgame-perfect Bayesian equilibrium. This is done by reducing the persuasion problem into a market segmentation problem. 

While our basic model assumes binary action and state space, we show that a similar approach can solve a more general class of product adoption two-stage persuasion problems, with an arbitrary state space. We provide sufficient conditions on the structure of the problem (namely, on the receiver utilities) that guarantee the existence of a reduction to the market segmentation problem that yields an optimal solution for the persuasion platforms problem. While this approach turns out to be useful even in general product adoption problems (namely, when the receiver's actions are "adopt" and "opt-out"), characterizing the solution of a more general two-stage persuasion problem (i.e., more than two actions for the receiver), is out of scope for this paper and left as an interesting open question.

Following the one-shot case, we introduced the reputation-based persuasion platform problem in which myopic users arrive sequentially and analyzed the equilibrium behavior of the platform and sender in this setting, and simplified the optimization problem of the platform to a finite dimension optimization problem. However, a complete solution to the platform's optimization problem in the reputation-based setting is left for future work.

We provide a numerical example in which the sender's utility and the distribution of users on which the sender is forced to act truthfully are approximated using an exhaustive search, which is tractable in the binary user type case. Our simulations suggest that the platform's optimal strategy favors users who are "harder" for the seller in the first place. This imbalance is more prevalent as the sender is less patient. That is, while this reputation mechanism indeed benefits the users on average, not all users benefit equally. This insight calls for further research quantifying and controlling the degree of fairness within this reputation-based framework.

Overall, our results suggest that introducing a platform to control the information available to the sender can incentivize it to take a more truthful strategy and protect users from being recommended low-quality products. These findings have potential applications in economic situations in which the seller is provided with information regarding user preferences by some third-party entity, which is interested in maximizing users' welfare. Importantly, several aspects of some real-life applications of reputation systems are out of the scope of this work. This includes, for instance, the consideration of long-living users (\citealp{romanyuk2019cream,lorecchio2023bad}) and the effect of reputation on social learning (\citealp{che2018recommender,acemoglu2022learning}). Such endeavors could complement our findings, leading to a richer understanding of how reputation and recommendation systems operate and evolve.

\newpage
\bibliographystyle{plainnat}
\bibliography{sample}

\newpage
\appendix
\section{Omitted Proofs}

\subsection{Proof of Lemma \ref{lem2}}

\begin{proof}
    Given a platform policy $\sigma$, note that a deviation from $p_T$ can occur only conditional on $x=x^T$. It is clear that such deviation should be of the form $p(x^T) = p_k$ for some $k \in [n]$. Playing $p_T$ always leaves the sender at the high reputation state, hence its expected utility from the current timestep is $\mu$ and the expected utility from all future timesteps is $V(\sigma)$. 
    
    On the other hand, while deviating to $p_k$ the expected utility from the current timestep is $F_k(x^T)(\mu + (1-\mu)p_k)$ and the expected utility from all future timesteps is $\bar{u}$ with probability $(1-\mu)F_k(x^T)p_k$ (which corresponds to the product being of low quality, the sender decides to recommend purchasing it and the drawn user has persuasion threshold above $p_k$), and $V(\sigma)$ otherwise. 
    
    Therefore, the following inequality reflects the fact that deviation from $p_T$ to $p_k$ conditional on $x = x^T$ is not beneficial for the sender:

    \begingroup
    \allowdisplaybreaks
    \begin{align}
    (1-\delta)\mu +\delta V(\sigma) &\ge (1-\delta)F_k(x^T)(\mu + (1-\mu)p_k)  \nonumber \\
    &\quad+ \delta (1-\mu) F_k(x^T) p_k \bar{u} + \delta (1 - (1-\mu) F_k(x^T) p_k) V(\sigma) \Leftrightarrow  \nonumber \\
    \delta (1-\mu) F_k(x^T) p_k V(\sigma) &\ge (1-\delta)(F_k(x^T) \mu - \mu + F_k(x^T)(1-\mu)p_k) \nonumber \\
    &\quad+ \delta (1-\mu) F_k(x^T) p_k \bar{u} \Leftrightarrow \nonumber \\
    V(\sigma) &\ge \frac{(1-\delta)(F_k(x^T) \mu - \mu + F_k(x^T)(1-\mu)p_k) + \delta (1-\mu) F_k(x^T) p_k \bar{u}}{\delta (1-\mu) F_k(x^T) p_k} \nonumber \\
    &= \frac{(1-\delta)(F_k(x^T) \mu - \mu + F_k(x^T)(1-\mu)p_k)}{\delta (1-\mu) F_k(x^T) p_k} + \bar{u} \nonumber \\
    &= \frac{1-\delta}{\delta} \cdot \left(\frac{\mu}{1-\mu} \cdot \frac{F_k(x^T) - 1}{F_k(x^T) p_k} + 1\right) + \bar{u} \nonumber \\
    &= \frac{1-\delta}{\delta} \cdot \left(1 - \frac{1 - F_k(x^T)}{F_k(x^T) \theta_k}\right) + \bar{u} \nonumber
    \end{align}
    \endgroup

    When the above inequality holds for every $k \in [n]$, $p_T$ is indeed a best-response of the sender to the platform policy $\sigma$.
\end{proof}

\subsection{Proof of Lemma \ref{lem3}}

\begin{proof}
    Let $\sigma_\epsilon \in \Sigma_\epsilon$. From Bayes-plausibility of $\sigma$ and $\sigma_\epsilon$ we get:
    \begin{gather*}
        x^* = \alpha^T x^T + (1-\alpha^T)x^F \\
        x^* = (1-\epsilon)\alpha^T x^T + (1-(1-\epsilon)\alpha^T)x^F_\epsilon
    \end{gather*}
    Combining the two equations and isolating $x^F_\epsilon$, we get:
    \begin{gather}
        (1-(1-\epsilon)\alpha^T)x^F_\epsilon = (\alpha^T - (1-\epsilon)\alpha^T) x^T + (1-\alpha^T)x^F = \epsilon \alpha^T x^T + (1-\alpha^T)x^F \nonumber \\
        x^F_\epsilon = \frac{\epsilon \alpha^T x^T + (1-\alpha^T)x^F}{1-(1-\epsilon)\alpha^T} \to \frac{(1-\alpha^T)x^F}{1-\alpha^T} = x^F \nonumber
    \end{gather}
    when $\epsilon \to 0$.
\end{proof}

\subsection{Proof of Lemma \ref{lem4}}

\begin{proof}
    For any $\epsilon > 0$, denote by $\sigma_\epsilon$ the unique platform policy that satisfies $\sigma_\epsilon \in \Sigma_\epsilon$ and $\sigma^F_\epsilon = BBM(x^F_\epsilon)$, where $BBM(x)$ is the policy obtained by applying the algorithm of \cite{bergemann2015} with respect to the prior user distribution $x$. 
    As shown in \cite{bergemann2015}, this policy satisfies
    \begin{gather*}
        (1)\ U_S(p^*; x^F_\epsilon)=U_S(\sigma^F_\epsilon, p^*)
    \end{gather*}
    We also know that,
    \begin{gather*}
        (2)\ U_S(p^*; x^F)<U_S(\sigma_F, p^*)
    \end{gather*}
    From the assumption on $\sigma$. In addition, $U_S(p^*;x)$ is continuous in $x$, therefore from Lemma \ref{lem3} we know that,
    \begin{gather*}
        (3)\ U_S(p^*; x^F_\epsilon) \to U_S(p^*; x^F)
    \end{gather*}
    when taking $\epsilon \to 0$. Now, taking $\epsilon \to 0$ yields:
    \begin{gather*}
        V(\sigma_\epsilon) = (1-\epsilon)\alpha^T \mu + (1-(1-\epsilon)\alpha^T) U_S(\sigma^F_\epsilon, p^*) \underset{(1)}{=} \\
        (1-\epsilon)\alpha^T \mu + (1-(1-\epsilon)\alpha^T) U_S(p^*; x^F_\epsilon) \underset{(3)}{\to}
        \alpha^T \mu + (1-\alpha^T) U_S(p^*; x^F) \\ \underset{(2)}{<} \alpha^T \mu + (1-\alpha^T) U_S(\sigma_F, p^*) = V(\sigma)
    \end{gather*}
    Therefore, for small enough $\epsilon^* > 0$, the policy $\sigma_l = \sigma_{\epsilon^*}$ as defined above satisfies $V(\sigma_l) < V(\sigma)$.
\end{proof}

\subsection{Proof of Lemma \ref{lem5}}

\begin{proof}
    Let $0 < \epsilon < 1$. Define $\sigma'$ and $\sigma''$ as follows:
    \begingroup
    \allowdisplaybreaks
    \begin{gather*}
    \operatorname{supp}(\sigma') = \{ x'^1, ... x'^n, x'^T \} \\
    x'^T = x^T, \alpha'^T = \alpha^T \\
    \forall i, j \in [n]: x'^i_j = \ind{i=j}, \alpha'^i=(1-\alpha^T)x^F_i \\
    \operatorname{supp}(\sigma'') = \{ x''^1, ... x''^n, x''^T \} \\
    x''^T = x^T, \alpha''^T = (1-\epsilon)\alpha^T \\
    \forall i, j \in [n]: x''^i_j = \ind{i=j}, \alpha''^i=(1-\alpha^T+\epsilon \alpha^T)x^F_{\epsilon,i}
    \end{gather*}
    \endgroup
    
    By construction, $\sigma' \in \Sigma$ and $\sigma'' \in \Sigma_{\epsilon}$. First, let us show that $V(\sigma') < V(\sigma'')$. Start by defining:
    \begin{gather*}
        \lambda \coloneqq \frac{\epsilon \alpha^T}{1 - \alpha^T + \epsilon \alpha^T} \\
        \forall j \in [n]: u_j \coloneqq \mu + (1-\mu) p_j \\
        W_F \coloneqq \sum_{j=1}^n x^F_j u_j, W_T \coloneqq \sum_{j=1}^n x^T_j u_j
    \end{gather*}
    Notice that since $\sigma'' \in \Sigma_{\epsilon}$, $x''^F = x^F_\epsilon = \lambda x^T + (1-\lambda) x^F$. It is straightforward to show that $U_S(\sigma'_F,p^*) = W_F$. We next express $U_S(\sigma''_F,p^*)$ in terms of $\lambda, W_F$ and $W_T$:
    \begin{align*}
        U_S(\sigma''_F,p^*) &= \sum_{j=1}^n x^F_{\epsilon,j} u_j = \sum_{j=1}^n (\lambda x^T + (1-\lambda) x^F)_j u_j \\
        &= \lambda \sum_{j=1}^n x^T_j u_j + (1-\lambda)\sum_{j=1}^n x^F_j u_j = \lambda W_T + (1-\lambda) W_F
    \end{align*}
    Now,

    \begingroup
    \allowdisplaybreaks
    \begin{align*}
        V(\sigma'') > V(\sigma') &\Leftrightarrow 
        \alpha''^T \mu + (1-\alpha''^T) U_S(\sigma''_F, p^*) > \alpha'^T \mu + (1-\alpha'^T) U_S(\sigma'_F, p^*) \\
        &\Leftrightarrow (1-\epsilon)\alpha^T \mu + (1-(1-\epsilon)\alpha^T) U_S(\sigma''_F, p^*) > \alpha^T \mu + (1-\alpha^T) U_S(\sigma'_F, p^*) \\
        &\Leftrightarrow (1-\alpha^T+\epsilon\alpha^T) U_S(\sigma''_F, p^*) > \epsilon \alpha^T \mu + (1-\alpha^T) U_S(\sigma'_F, p^*) \\
        &\Leftrightarrow (1-\alpha^T+\epsilon\alpha^T) (\lambda W_T + (1-\lambda) W_F) > \epsilon \alpha^T \mu + (1-\alpha^T) W_F \\
        &\Leftrightarrow \epsilon \alpha^T W_T + (1-\alpha^T) W_F > \epsilon \alpha^T \mu + (1-\alpha^T) W_F
    \end{align*}
    \endgroup
    
    which indeed holds because $W_T > \mu$. Next, notice that $V(\sigma) \le V(\sigma')$ follows directly from \cite{blackwell}, as $\sigma_F$ is a refinement of $\sigma'_F$. Therefore, we overall get $V(\sigma) \le V(\sigma') < V(\sigma'')$, so setting $\sigma_h = \sigma''$ completes the proof.
\end{proof}

\subsection{Proof of Lemma \ref{lem6}}

\begin{proof}
    First, the fact that $\sigma$ is lowest-type-targeting implies that $I_j(x^i,p^*) = x_j^i$ for any $i \in [m], j \in [n]$ (since $x_j^i > 0$ implies $\ind{p^*(x^i) \le p_j} = 1$). Now, we get that:
    \begingroup
    \allowdisplaybreaks
    \begin{align*}
        U_S(\sigma_F, p^*) &= \sum_{i=1}^m \sigma_F(x^i) \sum_{j=1}^n I_j(x^i,p^*) (\mu + (1-\mu)p^*(x^i)) \\ &= \sum_{i=1}^m \frac{\alpha^i}{1-\alpha^T} \sum_{j=1}^n x^i_j (\mu + (1-\mu)p^*(x^i)) \\& =
        \frac{1}{1-\alpha^T} \sum_{i=1}^m \alpha^i (\mu + (1-\mu)p^*(x^i)) \sum_{j=1}^n x^i_j \\&= \frac{1}{1-\alpha^T} \sum_{i=1}^m \alpha^i (\mu + (1-\mu)p^*(x^i)) \\&= 
        \mu \frac{1}{1-\alpha^T} \sum_{i=1}^m \alpha^i + (1-\mu) \frac{1}{1-\alpha^T} \sum_{i=1}^m \alpha^i p^*(x^i) \\&= \mu + \frac{1-\mu}{1-\alpha^T} \sum_{i=1}^m \alpha^i p^*(x^i) \\
        U_P(\sigma_F, p^*) &= \sum_{i=1}^m \sigma_F(x^i) \sum_{j=1}^n I_j(x^i,p^*) (\mu\theta_j - (1-\mu)p^*(x^i)) \\&= \sum_{i=1}^m \frac{\alpha^i}{1-\alpha^T} \sum_{j=1}^n x^i_j (\mu\theta_j - (1-\mu)p^*(x^i)) \\&= 
        \frac{1}{1-\alpha^T} \sum_{i=1}^m \alpha^i \sum_{j=1}^n x^i_j (\mu\theta_j - (1-\mu)p^*(x^i)) \\ &= \frac{\mu}{1-\alpha^T} \sum_{i=1}^m \alpha^i (\sum_{j=1}^n x^i_j \theta_j) - \frac{1-\mu}{1-\alpha^T} \sum_{i=1}^m \alpha^i p^*(x^i) \sum_{j=1}^n x^i_j \\ &=
        \frac{\mu}{1-\alpha^T} \sum_{i=1}^m \alpha^i (\sum_{j=1}^n x^i_j \theta_j) - \frac{1-\mu}{1-\alpha^T} \sum_{i=1}^m \alpha^i p^*(x^i)
    \end{align*}
    \endgroup

    Now, plugging into $U_S(\sigma, p_T)$ and $U_P(\sigma, p_T)$, we get:

    \begingroup
    \allowdisplaybreaks
    \begin{align*}
        U_S(\sigma, p_T) &= \alpha^T \mu + (1-\alpha^T) U_S(\sigma_F,p^*) \\&= \alpha^T \mu + (1-\alpha^T) \mu + (1-\alpha^T) \frac{1-\mu}{1-\alpha^T} \sum_{i=1}^m \alpha^i p^*(x^i) \\&= \mu + (1-\mu) \sum_{i=1}^m \alpha^i p^*(x^i) \\
        U_P(\sigma, p_T) &= \alpha^T \mu (\sum_{j=1}^n x^T_j \theta_j) + (1-\alpha^T) U_P(\sigma_F,p^*) \\&= \alpha^T \mu (\sum_{j=1}^n x^T_j \theta_j) + (1-\alpha^T) \frac{\mu}{1-\alpha^T} \sum_{i=1}^m \alpha^i (\sum_{j=1}^n x^i_j \theta_j) - (1-\alpha^T) \frac{1-\mu}{1-\alpha^T} \sum_{i=1}^m \alpha^i p^*(x^i) \\&=
        \mu (\alpha^T \sum_{j=1}^n x^T_j \theta_j + \sum_{i=1}^m \alpha^i (\sum_{j=1}^n x^i_j \theta_j)) - (1-\mu) \sum_{i=1}^m \alpha^i p^*(x^i) \\ &= 
        \mu \sum_{j=1}^n x^*_j \theta_j - (1-\mu) \sum_{i=1}^m \alpha^i p^*(x^i)
    \end{align*}
    \endgroup
    where the last equality comes from Bayes-plausibility.
\end{proof}

\subsection{Proof of Lemma \ref{lem7}}

\begin{proof}
    Let $\sigma \in \Sigma$ be some platform policy which is not lowest-type-targeting, i.e., there exists $x \in \operatorname{supp}(\sigma_F)$ such that $p^*(x) = \frac{\mu}{1-\mu} \theta^x_{opt} > \frac{\mu}{1-\mu} \theta^x_{min}$. Denote by $i$ and $j$ the indices satisfying $\theta_i = \theta^x_{min}$ and $\theta_j = \theta^x_{opt}$, and notice that $i<j$. Now, define a new policy $\tilde{\sigma}$ which is similar to $\sigma$, except $x$ is decomposed into $y$ and $z$ as follows:

    \begin{align*}
        y &= \frac{1}{F_j(x)} (0, ... 0, x_j, x_{j+1}, ... x_n) \\
        z &= \frac{1}{1-F_j(x)} (x_1, ... x_{j-1}, 0, ... 0) \\
        \tilde{\sigma}(y) &= F_j(x) \sigma(x), \tilde{\sigma}(z) = (1 - F_j(x)) \sigma(x)
    \end{align*}

    Notice that $\tilde{\sigma}$ is Bayes-plausible by construction. Clearly $x^T = \tilde{x}^T$, $\alpha^T = \tilde{\alpha}^T$. Now, notice that $p^*(z) < p^*(x)$ (as $\theta^x_{opt} > \max \operatorname{supp}(z)$) and $p^*(y) = p^*(x)$:

    \begingroup
    \allowdisplaybreaks
    \begin{align*}
    j &= \argmax_{l=1, ... n} \sum_{k=l}^n x_k (\mu + (1-\mu)p_l) =
    \argmax_{l=j, ... n} \sum_{k=l}^n x_k (\mu + (1-\mu)p_l) \\&= 
    \argmax_{l=j, ... n} \frac{1}{F_j(x)} \sum_{k=l}^n x_k (\mu + (1-\mu)p_l) = 
    \argmax_{l=j, ... n} \sum_{k=l}^n y_k (\mu + (1-\mu)p_l) = j'
    \end{align*}
    \endgroup

    where $j'$ is the index for which $p^*(y) = \frac{\mu}{1-\mu} \theta_{j'}$. 
    
    Now, notice that users' utility weakly increases in $\tilde{\sigma}$ with respect to $\sigma$, since users moved from $x$ to $z$ can only benefit while users moved from $x$ to $y$ remain with the same utility. Therefore, $U_P(\sigma, p_T) \le U_P(\tilde{\sigma}, p_T)$. As for the sender's utility, notice that:

    \begingroup
    \allowdisplaybreaks
    \begin{align*}
    &U_S(\tilde{\sigma},p_T) > U_S(\sigma,p_T) \Leftrightarrow U_S(\tilde{\sigma}_F,p^*) > U_S(\sigma_F,p^*) \Leftrightarrow \\
    &\begin{aligned}
    &\tilde{\sigma}(y) \sum_{k=1}^n y_k \ind{p^*(y) \le p_k} (\mu + (1-\mu)p^*(y)) + \tilde{\sigma}(z) \sum_{k=1}^n z_k \ind{p^*(z) \le p_k} (\mu + (1-\mu)p^*(z)) \\
    &> \sigma(x) \sum_{k=1}^n x_k \ind{p^*(x) \le p_k} (\mu + (1-\mu)p^*(x)) \Leftrightarrow
    \end{aligned} \\
    &\begin{aligned}
    &F_j(x) \sum_{k=1}^n y_k \ind{p^*(y) \le p_k} (\mu + (1-\mu)p^*(y)) + (1-F_j(x)) \sum_{k=1}^n z_k \ind{p^*(z) \le p_k} (\mu + (1-\mu)p^*(z)) \\
    &> \sum_{k=1}^n x_k \ind{p^*(x) \le p_k} (\mu + (1-\mu)p^*(x)) \Leftrightarrow
    \end{aligned} \\
    &\begin{aligned}
    &\sum_{k=j}^n x_k \ind{p^*(x) \le p_k} (\mu + (1-\mu)p^*(x)) + (1-F_j(x)) \sum_{k=1}^n z_k \ind{p^*(z) \le p_k} (\mu + (1-\mu)p^*(z)) \\
    &> \sum_{k=1}^n x_k \ind{p^*(x) \le p_k} (\mu + (1-\mu)p^*(x)) \Leftrightarrow
    \end{aligned} \\
    &\begin{aligned}
    &(1-F_j(x)) \sum_{k=1}^n z_k \ind{p^*(z) \le p_k} (\mu + (1-\mu)p^*(z)) 
    > \sum_{k=1}^{j-1} x_k \ind{p^*(x) \le p_k} (\mu + (1-\mu)p^*(x)) \Leftrightarrow
    \end{aligned} \\
    &\begin{aligned}
    &\sum_{k=1}^{j-1} x_k \ind{p^*(z) \le p_k} (\mu + (1-\mu)p^*(z))
    > \sum_{k=1}^{j-1} x_k \ind{p^*(x) \le p_k} (\mu + (1-\mu)p^*(x))
    \end{aligned}
    \end{align*}
    \endgroup
    
Now, notice that the left-hand side is strictly greater than zero since $p^*(z)$ is the optimal sender strategy at segment $z$, and $z = (x_1, ... x_{j-1}, 0 ... 0)$ up to the constant $(1-F_j(x))$. On the other hand, the right-hand side equals zero since for all $k \in [j-1]: p^*(x) > p_k$. Therefore the condition holds and we get $U_S(\tilde{\sigma},p_T) > U_S(\sigma,p_T)$.
    
\end{proof}

\subsection{Proof of Lemma \ref{lem8}}

\begin{proof}
    Let $\sigma \in \Sigma^*$ be some optimal platform policy. If $\sigma$ is lowest-type-targeting then we are done. Otherwise, from Lemma \ref{lem7} there exists a platform policy $\tilde{\sigma}$ such that $x^T = \tilde{x}^T$, $\alpha^T = \tilde{\alpha}^T$, $U_P(\tilde{\sigma},p_T) \ge U_P(\sigma,p_T)$ and $U_S(\tilde{\sigma},p_T) > U_S(\sigma,p_T)$. From Corollary \ref{cor5} combined with the fact that $\sigma$ is incentive-compatible, it follows that $\tilde{\sigma}$ is also incentive-compatible. Therefore it follows that $U_P(\tilde{\sigma},p_T) = U_P(\sigma,p_T)$, otherwise it contradicts the fact that $\sigma \in \Sigma^*$, hence  $\tilde{\sigma} \in \Sigma^*$. Now, if $\tilde{\sigma}$ is lowest-type-targeting then we're done, otherwise, we repeat the same process. Note that this process can only be repeated a finite number of times, therefore we must end up with an optimal lowest-type-targeting platform policy.
\end{proof}

\subsection{Proof of Lemma \ref{lem9}}

\begin{proof}
    Lemma \ref{lem7} implies that any platform policy which is not lowest-type-targeting is also not Pareto-efficient. To show the opposite direction, we now show that the sum of the sender and platform utility of any platform policy is maximal for lowest-type-targeting policies:
    From Lemma \ref{lem6}, the sum of utilities for any lowest-type-targeting $\tilde{\sigma}$ is:
    \begin{equation*}
        U_P(\tilde{\sigma}, p_T) + U_S(\tilde{\sigma}, p_T) = \mu (1 + \sum_{j=1}^n x^*_j \theta_j)
    \end{equation*}
    
    And for any $\sigma \in \Sigma$ the sum of utilities is:
    \begin{align*}
        U_P(\sigma, p_T) + U_S(\sigma, p_T) &= \alpha^T \mu (1 + \sum_{j=1}^n x^T_j \theta_j) + \sum_{i=1}^m \alpha^i (\sum_{j=1}^n I_j(x^i, p^*) (1 + \theta_j)) \\
        &\le \alpha^T \mu (1 + \sum_{j=1}^n x^T_j \theta_j) + \sum_{i=1}^m \alpha^i (\sum_{j=1}^n (1 + \theta_j)) \\
        &= \mu + \alpha^T \mu (\sum_{j=1}^n x^T_j \theta_j) + \sum_{i=1}^m \alpha^i (\sum_{j=1}^n (\theta_j)) = \mu (1 + \sum_{j=1}^n x^*_j \theta_j)
    \end{align*}
    where the last equality comes from Bayes-plausibility. Now, assume by contradiction that there exists a lowest-type-targeting policy that is not Pareto-efficient. Then there exists some other policy such that the sum of utilities is strictly larger, in contradiction to the claim we just proved.
\end{proof}

\end{document}